\documentclass[journal,twocolumn]{IEEEtran}
\usepackage{amsfonts}
\usepackage{amsmath,amssymb}
\usepackage{graphicx,subfigure}
\usepackage{color}      
\usepackage{epsfig}
\usepackage{ulem}
\usepackage{tipa}
\usepackage{bbm}
\def\ep{\epsilon_n}
\def\pm{P'_{1\rightarrow 2}}
\def\pn{P'_{2\rightarrow 1}}

\newcommand{\hG} {\hat{G}}
\newcommand{\hV} {\hat{V}}

\def\endproof{\hspace*{\fill}~$\blacksquare$}
\long\def\comment#1{} 

\newtheorem{theorem}{Theorem}
\newtheorem{lemma}[theorem]{Lemma}
\newtheorem{claim}[theorem]{Claim}

\newtheorem{corollary}[theorem]{Corollary}
\newtheorem{definition}[theorem]{Definition}

\begin{document}

\title{On the multiple unicast capacity of 3-source, 3-terminal directed acyclic networks}

\author{Shurui Huang {\it Student Member, IEEE} and Aditya Ramamoorthy, {\it Member, IEEE}
\thanks{The authors are with the Dept. of Electrical \& Computer Eng., Iowa State Univ., Ames, IA 5011. The material in this paper has appeared in part at the IEEE Information Theory and Applications Workshop (ITA 2012), San Diego, CA, Feb. 5 - 10, 2012, and the IEEE International Conference on Communications (ICC 2011), Kyoto, Japan, Jun. 5 - 9, 2011. This work was supported in part by NSF grant CCF-1018148.}}


\maketitle 
\begin{abstract} 
We consider the multiple unicast problem with three source-terminal pairs over directed acyclic networks with unit-capacity edges. The three $s_i-t_i$ pairs wish to communicate at unit-rate via network coding. The connectivity between the $s_i - t_i$ pairs is quantified by means of a connectivity level vector, $[k_1~k_2~k_3]$ such that there exist $k_i$ edge-disjoint paths between $s_i$ and $t_i$. In this work we attempt to classify networks based on the connectivity level.
It can be observed that unit-rate transmission can be supported by routing if $k_i \geq 3$, for all $i = 1, \dots, 3$.  In this work, we consider, connectivity level vectors such that $\min_{i = 1, \dots, 3} k_i  < 3$.  We present either a constructive linear network coding scheme or an instance of a network that cannot support the desired unit-rate requirement, for all such connectivity level vectors except the vector $[1~2~4]$ (and its permutations).
The benefits of our schemes extend to networks with higher and potentially different edge capacities. Specifically, our experimental results indicate that for networks where the different source-terminal paths have a significant overlap, our constructive unit-rate schemes can be packed along with routing to provide higher throughput as compared to a pure routing approach.

\end{abstract}
\section{Introduction} 
In a network that supports multiple unicast, there are several source terminal pairs; each source
wishes to communicate with its corresponding terminal. Multiple unicast connections form bulk of the traffic over both wired and wireless networks. Thus, network coding schemes that
can help improve network throughput for multiple unicasts are of considerable interest. However, it is well recognized that the design of constructive network coding schemes for multiple unicasts
is a hard problem when compared with the case of multicast that is very well understood \cite{rm, hoMKKESL06, ebrahimiF11}. Specifically, it is known that there are instances of networks where linear (whether scalar or vector) network coding is insufficient \cite{DoughertyFZ05}.

The multiple unicast problem has been examined for both directed
acyclic networks \cite{medardEHK03,yan06isit,HarveyIT,Traskov06, tracy06allerton,wangIT10,feder09,Kamath11,javidi08,JafarISIT,interfernce10}
and undirected networks \cite{LiLiunicast} in previous work.

The work of \cite{yan06isit}, provides an information theoretic
characterization for directed acyclic networks. However, this bound is not computable. The work of \cite{HarveyIT} proposes an outer bound for general directed networks. However, this bound is hard to evaluate even for small networks due to the large number of inequalities involved. There have been attempts to find constructive schemes leveraging network coding between pairs of sources \cite{Traskov06,tracy06allerton}. Numerous works consider restricted cases such as unicast with two sessions \cite{wangIT10,feder09, Kamath11,javidi08} and unicast with three sessions \cite{JafarISIT,interfernce10, wang11}. We discuss the related work in detail in Section \ref{sec:background}.

In this work we consider network coding for wired three-source, three-terminal directed acyclic networks with unit capacity edges. There are source-terminal pairs denoted $s_i-t_i, i = 1,\dots,3$, such that the maximum flow from $s_i$ to $t_i$ is $k_i$. Each source contains a unit-entropy message that needs to be communicated to the
corresponding terminal. 
In this work, for a given connectivity level vector $[k_1~k_2~k_3]$ we attempt to either design a constructive scheme based on linear network codes or demonstrate an instance of a network where supporting unit-rate transmission is impossible. Our achievability schemes use a combination of  random linear network coding and appropriate precoding. 
Our solutions are based on either scalar or vector network codes that operate over at most two time units (i.e., two network uses). This is useful, as one can arrive at multiple unicast schemes for arbitrary rates by packing unit-rate structures for which our achievability schemes apply.

%
%
%
\noindent \uline{{\it Main Contributions}}
\begin{list}{}{\leftmargin=0.1in \labelwidth=0cm \labelsep = 0cm}
\item[$\bullet~$] For the case of three unicast sessions with unit rates, we
identify certain feasible and infeasible connectivity levels $[k_1~k_2
~k_3]$. For the feasible cases, we construct schemes based on linear
network coding. For the infeasible cases, we provide
counter-examples, i.e., instances of graphs where the multiple
unicast cannot be supported under any (potentially nonlinear)
network coding scheme.
\item[$\bullet~$] We provide experimental results that demonstrate that our feasible schemes for unit-rate are useful for networks with higher capacity edges.
Specifically, we demonstrate classes of networks with higher capacity edges, where {\it packing} our unit-rate schemes allows us to achieve transmission rates that are strictly greater than those achieved by pure routing.
\end{list}
This paper is organized as follows. Section \ref{sec:background} contains an overview of related work. In Section \ref{sec:pre}, we introduce
the network coding model and problem formulation. Section \ref{sec:netcod_infeas} discusses infeasible instances, and Section \ref{sec:netcod_feas} discusses our achievable schemes for 3-source, 3-terminal multiple unicast networks. 
Section \ref{sec:simu} presents simulation results on networks with higher capacity edges and
Section \ref{sec:con} concludes the paper with a discussion of future work.

\section{Background and Related Work}
\label{sec:background}

It is well-recognized that network coding for multiple unicast is significantly harder than the network coding for multicast. The work of \cite{rm} establishes an equivalence between network coded multicast and the problem of solving systems of linear equations. In the same paper, they also point out that for multiple unicast, one also needs to somehow decode the intended message in the presence of undesired interference. In general, it is intractable to find network code assignments that simultaneously allow the intended message to be decoded while mitigating the interference. In fact, it is known that linear codes are insufficient for the multiple unicast problem \cite{DoughertyFZ05}.

In this work our focus is exclusively on multiple unicast for directed acyclic networks (see \cite{LiLiunicast} for the undirected case). Previous work in this domain includes the work of \cite{yan06isit} that presents an information theoretic characterization of the capacity region. However, in practice this bound is not computable due to the lack of upper bounds on the cardinality of the alphabets of the random variables involved in the characterization. Moreover, even for small sized networks, the number of inequalities involved is very large. Similar issues exist with the outer bound of \cite{HarveyIT}. There have been numerous works on achievable schemes for multiple unicast. The butterfly network with two unicast sessions is an instance where there is clear advantage to performing network coding over routing. Accordingly Traskov et al. \cite{Traskov06} proceed by packing butterfly networks for general multiple unicast. Ho et al. \cite{tracy06allerton} propose an achievable region by using XOR coding coupled with back-pressure algorithms. Multiple unicast in the presence of link faults and errors, under certain restricted (though realistic) network topologies has been studied in \cite{kamalRLL11}\cite{shizhengR11}.

Further progress has been made in certain restricted classes of problems. For instance, an improved outer bound (GNS bound) over the network sharing outer bound for two-unicast is proposed in \cite{Kamath11}. Price et al. \cite{javidi08} also propose an outer bound for two-unicast and demonstrate a network for which the outer bound is the exact capacity region. For two-unicast, Wang et al. \cite{wangIT10} (also see \cite{shenvi}) present a necessary and sufficient condition for unit-rate transmission and the work of \cite{feder09} and \cite{huangR_ITW11} propose an achievable region for general rates.

Some recent work deals with the case of three unicast sessions, which is also the focus of our work. The work of \cite{JafarISIT} and \cite{interfernce10} use the technique of interference alignment (proposed in \cite{caja08}) for multiple unicast. Roughly speaking they use random linear network coding and design appropriate precoding matrices at the source nodes that allow undesired interference at a terminal to be aligned. However, their approach requires several algebraic conditions to be satisfied in the network. It does not appear that these conditions can be checked efficiently. There has been a deeper investigation of these conditions in \cite{wang11}.
Our work is closest in spirit to these papers. Specifically, we also examine network coding for the three-unicast problem. However, the problem setting is somewhat different. Considering networks with unit capacity edges and given the maximum-flow $k_i$ between each source ($s_i$) - terminal ($t_i$) pair we attempt to either design a network code that allows unit-rate communication between each source-terminal pair, or demonstrate an instance of a network where unit-rate communication is impossible. Our achievability schemes for unit rate are useful since they can be packed into networks with higher capacity edges. Furthermore, these schemes require vector network coding over at most two time units, unlike the work of \cite{JafarISIT} and \cite{interfernce10}, that require a significantly higher level of time-expansion.

\section{Preliminaries}
\label{sec:pre}
We represent the network as a directed acyclic graph $G = (V, E)$.
Each edge $e \in E$ has unit capacity and can transmit one symbol
from a finite field of size $q$ per unit time (we are free to
choose $q$ large enough). If a given edge has higher capacity, it
can be treated as multiple unit capacity edges. A directed edge
$e$ between nodes $i$ and $j$ is represented as $(i,j)$, so that
$head(e) = j$ and $tail(e) = i$. A path between two nodes $i$ and
$j$ is a sequence of edges $\{ e_1, e_2, \dots, e_k\}$ such that
$tail(e_1) = i, head(e_k) = j$ and $head(e_i) = tail(e_{i+1}), i =
1, \dots, k-1$. The network contains a set of $n$ source nodes
$s_i$ and $n$ terminal nodes $t_i, i = 1, \dots
n$. Each source node $s_i$ observes a discrete integer-entropy
source, that needs to be communicated to
terminal $t_i$. Without loss of generality, we assume that the
source (terminal) nodes do not have incoming (outgoing) edges. If
this is not the case one can always introduce an artificial source
(terminal) node connected to the original source (terminal) node
by an edge of sufficiently large capacity that has no incoming
(outgoing) edges.


We now discuss the network coding model under consideration in
this paper. For the sake of understanding the model, suppose for now that each source
has unit-entropy, denoted by $X_i$ (as will be evident, in the sequel we work with integer entropy sources). In scalar linear network coding, the
signal on an edge $(i,j)$ is a linear combination of the signals
on the incoming edges of $i$ or the source signals at $i$ (if $i$
is a source). We shall only be concerned with networks that are
directed acyclic and can therefore be treated as delay-free
networks \cite{rm}. Let $Y_{e_i}$ (such that $tail(e_i) = k$ and
$head(e_i) = l$) denote the signal on edge $e_i \in E$. Then, we
have
\begin{align*}
Y_{e_i} &= \sum_{\{e_j | head(e_j) = k\}} f_{j,i} Y_{e_j} \text{~if $k \in V \backslash \{s_1, \dots, s_n\}$}, \text{~and}\\
Y_{e_i} &= \sum_{j=1}^n a_{j,i} X_j
\text{~~ where $a_{j,i} = 0$ if $X_j$ is not observed at $k$.}
\end{align*}
The coefficients $a_{j,i}$ and $f_{j,i}$ are from the operational field.
Note that since the graph is directed acyclic, it is equivalently possible to
express $Y_{e_i}$ for an edge $e_i$ in terms of the sources
$X_j$'s. If $Y_{e_i} = \sum_{k=1}^n \beta_{e_i, k} X_k$ then we say that the
global coding vector of edge $e_i$ is $\boldsymbol{\beta}_{e_i} =
[ \beta_{e_i, 1} ~\cdots~ \beta_{e_i, n}]$. We shall also
occasionally use the term coding vector instead of global coding
vector in this paper. We say that a node $i$ (or edge $e_i$) is
downstream of another node $j$ (or edge $e_j$) if there exists a
path from $j$ (or $e_j$) to $i$ (or $e_i$).

Vector linear network coding is a generalization of the scalar
case, where we code across the source symbols in time, and the
intermediate nodes can implement more powerful operations.
Formally, suppose that the network is used over $T$ time units. We
treat this case as follows. Source node $s_i$ now observes a
vector source $[X_i^{(1)} ~ \dots ~ X_i^{(T)}]$. Each edge in the
original graph is replaced by $T$ parallel edges. In this graph,
suppose that a node $j$ has a set of $\beta_{inc}$ incoming edges
over which it receives a certain number of symbols, and $\beta_{out}$
outgoing edges. Under vector network coding, node $j$ chooses
a matrix of dimension $\beta_{out} \times \beta_{inc}$. Each row
of this matrix corresponds to the local coding vector of an
outgoing edge from $j$.

Note that the general multiple unicast problem, where edges have
different capacities and the sources have different entropies can
be cast in the above framework by splitting higher capacity edges
into parallel unit capacity edges and a higher entropy source into
multiple, collocated unit-entropy sources. This is the approach taken below.

An instance of the multiple unicast problem is specified by the
graph $G$ and the source terminal pairs $s_i - t_i,  i = 1, \dots,
n$, and is denoted $<G, \{s_i - t_i\}_{1}^n, \{R_i\}_{1}^n>$,
where the integer rates $R_i$ denote the entropy of the $i^{th}$ source.
The $s_i-t_i$ connections will be referred to as sessions that we need to support.

Let the sources at $s_i$ be denoted as $X_{i1}, \dots, X_{i R_i}$. The instance is said to have a scalar linear network coding
solution if there exist a set of linear encoding coefficients for
each node in $V$ such that each terminal $t_i$ can recover $X_{i1}, \dots, X_{i R_i}$
using the received symbols at its input edges. Likewise, it is
said to have a vector linear network coding solution with vector
length $T$ if the network employs vector linear network codes and
each terminal $t_i$ can recover $[X_{i1}^{(1)} ~ \dots ~ X_{i1}^{(T)}], \dots, [X_{i R_i}^{(1)} ~ \dots ~ X_{i R_i}^{(T)}] $. If the instance has either a scalar or a vector
network coding solution, we say that it is feasible.

We will also be interested in examining the existence of a routing
solution, wherever possible. In a routing solution, each edge
carries a copy of one of the sources, i.e., each coding vector is
such that at most one entry takes the value $1$, all others are
$0$. Scalar (vector) routing solutions can be defined in a
manner similar to scalar (vector) network codes. We now define
some quantities that shall be used throughout the paper.

\begin{definition}
{\it Connectivity level.} The connectivity level for
source-terminal pair $s_i - t_i$ is said to be $\beta$ if the maximum
flow between $s_i$ and $t_i$ in $G$ is $\beta$. The connectivity level
of the set of connections $s_1 - t_1, \dots, s_n - t_n$ is the
vector $[ \text{max-flow}(s_1 - t_1) ~ \text{max-flow}(s_2 - t_2)~
\dots ~ \text{max-flow}(s_n - t_n)]$.
\end{definition}

In this work our aim is to characterize the feasibility of the multiple unicast problem based on the connectivity level of the $s_i -t_i$ pairs. The questions that we seek to answer are of the following form - suppose that the connectivity level is $[k_1~ k_2~ \dots~ k_n]$. Does any instance always have a linear (scalar or vector) network coding solution? If not, is it possible to demonstrate a counter-example, i.e, an instance of a graph $G$ and $s_i - t_i$'s such that recovering the $i$-th source at $t_i$ for all $i$ is impossible under linear (or nonlinear) strategies?


We conclude this section by observing that a multiple unicast instance $<G, \{s_i - t_i\}_1^n, \{1, 1, \dots, 1\}>$ with connectivity level $[n ~n ~ \dots ~n]$ is always feasible. Let $X_i, i = 1, \dots, n$ denote the $i$-th unit entropy source. We employ vector routing over $n$ time units. Source $s_i$ observes
$[X_i^{(1)} ~ \dots ~ X_i^{(n)}]$ symbols. Each edge $e$ in the
original graph $G$ is replaced by $n$ parallel edges, $e^1, e^2,
\dots, e^n$. Let $G_{\alpha}$ represent the subgraph of this graph
consisting of edges with superscript $\alpha$. It is evident that
max-flow($s_{\alpha} - t_{\alpha}$) = $n$ over $G_{\alpha}$. Thus,
we transmit $X_{\alpha}^{(1)}, \dots, X_{\alpha}^{(n)}$ over $G_{\alpha}$
using routing, for all $\alpha=1,\dots,n$. It is clear that this strategy satisfies the
demands of all the terminals.
%
In general, though a network with the above connectivity level may not be able to support a scalar routing solution.


\section{Network coding for three unicast sessions - Infeasible Instances}
\label{sec:netcod_infeas}
It is clear based on the discussion above that for three unicast sessions if the connectivity level is $[3 ~3 ~3]$, then a vector
routing solution always exists. We investigate counter-examples for certain connectivity levels in this section.
\begin{lemma}
\label{lemma:meagerness_cases}
There exist multiple unicast instances with three unicast sessions, $<G, \{s_i - t_i\}_{i=1}^3, \{1, 1, 1\}>$ such that the connectivity levels $[2~2~2]$ and $[1~1~3]$ are infeasible.
\end{lemma}
\begin{proof}
The examples are shown in Figs. \ref{fig:3sources2paths} and \ref{fig:counter311paths}. In Fig. \ref{fig:3sources2paths}, the cut specified by the set of nodes $\{s_1,s_2,s_3,v_1,v_2\}$ has a value of two, while it needs to support a sum rate of three. Similarly in Fig. \ref{fig:counter311paths}, the cut $\{s_1,s_2,v_1\}$ has a value of one, but needs to support a rate of two.
\vspace{-0.1in}
\begin{figure}[htbp]
\center{
\subfigure[]{\label{fig:3sources2paths}
\includegraphics[width=33mm,clip=false, viewport=50 55 150 190]{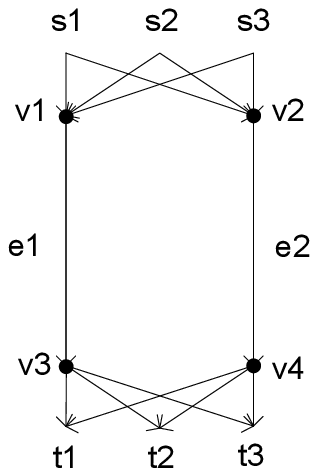}}
\hspace{0.4in}
\subfigure[]{\label{fig:counter311paths}
\includegraphics[width=40mm,clip=false, viewport=40 40 220 240]{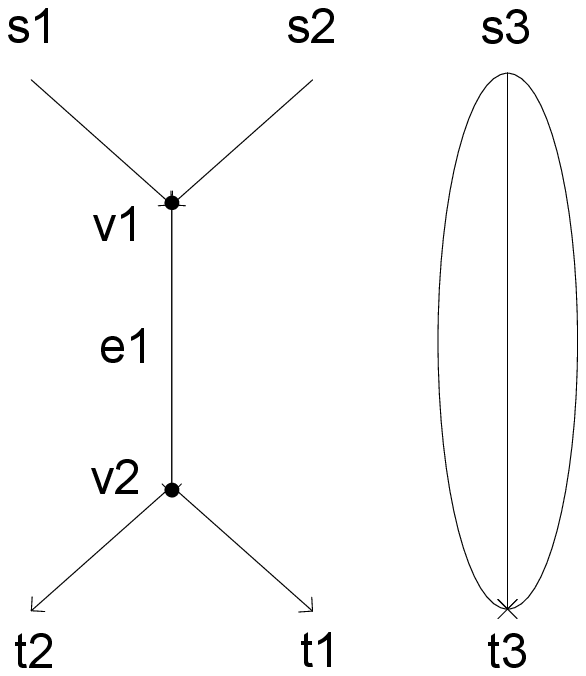}}
\caption{(a) An example of $[2~2~2]$ connectivity network without a network coding solution.
(b) An example of $[1~1~3]$ connectivity network without a network coding solution.}}
\end{figure}
\end{proof}
While the cutset bound is useful in the above cases, there exist certain connectivity levels for which a cut set bound is not tight enough.
We now present such an instance in Fig. \ref{fig:21counter}. 
This instance was also presented in 
\cite{feder09}, though the authors did not provide a formal proof of this fact. 

\begin{lemma}
\label{lemma:21counter}
There exists a multiple unicast instance, with two sessions $<G,\{s_1-t_1, s_2-t_2\}, \{2, 1\}>$ with connectivity level $[2~3]$ that is infeasible.
\end{lemma}
\begin{proof}
The graph instance is shown in Fig. \ref{fig:21counter}. Assume that in $n$ time units, $s_1$ observes two vector sources $[X_1^{(1)} ~ \dots ~ X_1^{(n)}]$ and $[X_2^{(1)} ~ \dots ~ X_2^{(n)}]$, $s_2$ observes one vector source $[X_3^{(1)} ~ \dots ~ X_3^{(n)}]$. The sources are denoted as $X_1^n$, $X_2^n$ and $X_3^n$ and are independent. The $n$ symbols that are transmitted over edge $(i,j)$ are denoted by $Y_{ij}^n$.
Suppose that the alphabet of $X_i$ is $\mathcal{X}$. Since the entropy rates for the three sources are the same,
we assume $H(X_i)= \log|\mathcal{X}|=a$. Also, since we are interested in
the feasibility of the solution, we assume that the alphabet size of $Y_{ij}$
is also the same as $\mathcal{X}$, and $H(Y_{ij})\leq\log|\mathcal{X}|=a$ by the capacity constraint of the edge.
At terminal $t_1$ and $t_2$, from $Y^n_{11}$, $Y^n_{12}$, $Y^n_{21}$ and $Y^n_{22}$,
we estimate $X^n_1$, $X^n_2$ and $X^n_3$.
Let the estimate be denoted as $\widehat{X}^n_1$, $\widehat{X}^n_2$ and $\widehat{X}^n_3$.
Suppose that there exist network codes and decoding functions such that
$P((\widehat{X}^n_1,\widehat{X}^n_2)\neq (X_1^n,X_2^n))\rightarrow 0$ as $n\rightarrow \infty$.
\begin{figure}[t]
\begin{center}
\includegraphics[width=40mm,clip=false, viewport=10 25 188 246]{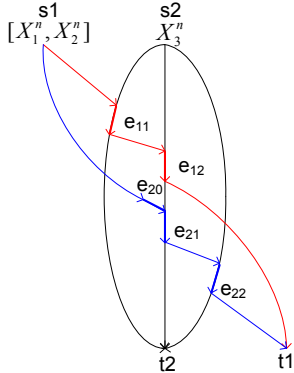}
\caption{An example of $[2~3]$ connectivity network, rate $\{2,1\}$
cannot be supported.} \label{fig:21counter}\centering
\end{center}
\end{figure}
For successful decoding at $t_1$, using Fano's inequality, we have
\begin{equation}
\label{eq:fano}
 H(X_1^n,X_2^n|\widehat{X}^n_1,\widehat{X}^n_2)\leq n\ep.
\end{equation}
where $n\ep=1+2nP_e\log(|\mathcal X|)$, $P_e=P((\widehat{X}^n_1,\widehat{X}^n_2)\neq (X_1^n,X_2^n))$ and $\ep \rightarrow 0$ as $n\rightarrow\infty$.
The topological structure of the network implies that $\widehat{X}^n_1,\widehat{X}^n_2$ are functions of $Y^n_{12}$ and $Y^n_{22}$. Hence, we have
\begin{equation}
\label{eq:fano2}
\begin{split}
H(X_1^n,X_2^n|Y^n_{12},Y^n_{22})&=H(X_1^n,X_2^n|\widehat{X}^n_1,\widehat{X}^n_2,Y^n_{12},Y^n_{22})\\
&\leq H(X_1^n,X_2^n|\widehat{X}^n_1,\widehat{X}^n_2)\leq n\ep.
\end{split}
\end{equation}
Since $H(Y^n_{12},Y^n_{22})\leq 2an$, using eq. (\ref{eq:fano2}) and the independence of $X^n_1$, $X^n_2$ and $X^n_3$, by Claim \ref{claim:endecode} (see Appendix), we have
\begin{align}
an-n\ep &\leq H(X_3^n|Y^n_{12},Y^n_{22})\leq an, \text{~and} \label{eq:cn}\\
&H(Y^n_{12},Y^n_{22}|X_3^n) \geq 2an-2n\ep. \label{eq:diff}
\end{align}
%
Next, we have
\begin{equation}
\label{eq:same}
\begin{split}
&H(Y^n_{21},Y^n_{22})\stackrel{_{(a)}}{=}H(X_3^n,Y^n_{21},Y^n_{22})-H(X_3^n|Y^n_{21},Y^n_{22})\\
&\stackrel{_{(b)}}{=}H(X_3^n,Y^n_{21})-H(X_3^n|Y^n_{21},Y^n_{22})\\
&\stackrel{_{(c)}}{\leq}2an-H(X_3^n|Y^n_{21},Y^n_{22},Y^n_{20},Y^n_{12},X_1^n,X_2^n)\\
&\stackrel{_{(d)}}{=}2an-H(X_3^n|Y^n_{22},Y^n_{20},Y^n_{12},X_1^n,X_2^n)\\
&\stackrel{_{(e)}}{=}2an-H(X_3^n|Y^n_{22},X_1^n,X_2^n,Y^n_{12})\\
&\stackrel{_{(f)}}{=}2an-H(X_3^n|Y^n_{22},Y^n_{12})+I(X_3^n;X_1^n,X_2^n|Y^n_{22},Y^n_{12})\\
&\leq 2an-H(X_3^n|Y^n_{22},Y^n_{12})+H(X_1^n,X_2^n|Y^n_{22},Y^n_{12})\\
&\stackrel{_{(g)}}{\leq}2an-an+n\ep+n\ep=an+2n\ep,
\end{split}
\end{equation}
where (a) follows from the chain rule, (b) holds because $Y^n_{22}$ is a
function of $X_3^n$ and $Y^n_{21}$, (c) follows from the capacity
constraints and the fact that conditioning reduces entropy, (d) follows as  $Y^n_{21}$ is a function of $Y^n_{12}$ and $Y^n_{20}$,
(e) is due to the fact that $Y^n_{20}$ is a function of $X_1^n$ and $X_2^n$, (f) 
follows from the definition of mutual information, and (g) is a consequence of
eq. (\ref{eq:fano2}) and eq. (\ref{eq:cn}). The above
inequalities indicate that $e_{21}$ and $e_{22}$ need to carry the
same information asymptotically for successful decoding at $t_1$.

From the network, we know that $Y^n_{12}$ is a function of
$Y^n_{11}$ and $X_3^n$. This implies that
\begin{equation}
\label{eq:indep}
\begin{split}
H(Y^n_{11},Y^n_{21},Y^n_{22}&|X_3^n)=H(Y^n_{11},Y^n_{21},Y^n_{22},X_3^n|X_3^n)\\
&\geq H(Y^n_{12},Y^n_{21},Y^n_{22}|X_3^n)\\
&\geq H(Y^n_{22},Y^n_{12}|X_3^n)\stackrel{_{(a)}}{\geq} 2an-2n\ep,
\end{split}
\end{equation}
where (a) is due to eq. (\ref{eq:diff}). Finally, we have
\begin{equation}
\begin{split}
&H(X_3^n|Y^n_{11},Y^n_{21},Y^n_{22})\\
&=H(Y^n_{11},Y^n_{21},Y^n_{22}|X_3^n)+H(X_3^n)-H(Y^n_{22},Y^n_{21},Y^n_{11})\\
&\stackrel{_{(a)}}{\geq}2an-2n\ep+an-H(Y^n_{22},Y^n_{21})-H(Y^n_{11}|Y^n_{22},Y^n_{21})\\
&\stackrel{_{(b)}}{\geq}3an-2n\ep-an-2n\ep-H(Y^n_{11}|Y^n_{22},Y^n_{21})\\
&\stackrel{_{(c)}}{\geq}2an-4n\ep-an=an-4n\ep,
\end{split}
\end{equation}
where (a) is due to eq. (\ref{eq:indep}), (b) is because of
eq. (\ref{eq:same}) and (c) holds because of the capacity constraint on $Y^n_{11}$. This implies that $t_2$ cannot decode $X_3^n$ with an asymptotically vanishing probability of error.
\end{proof}
\begin{corollary}
There exists a multiple unicast instance with three sessions, and connectivity level $[2~3~2]$ that is infeasible.
\end{corollary}
\begin{proof}
Consider the instance $<G,\{s'_i-t'_i\}^3_1,\{1,1,1\}>$, where $G$ is the graph in Fig. \ref{fig:21counter}.
The sources $s'_1$ and $s_3'$ are collocated at $s_1$ (in $G$), and the terminals $t_1'$ and $t_3'$ are collocated
at $t_1$ (in $G$). Likewise, the source $s_2'$ and terminal $t_2'$ are located at $s_2$ and $t_2$ in $G$.
The three sessions have connectivity level $[2~3~2]$. Based on the arguments in Lemma \ref{lemma:21counter},
there is no feasible solution for this instance.
\end{proof}

The previous example can be generalized to an instance with two unicast sessions with connectivity level $[n_1~n_2]$ that cannot support rates $R_1 = n_1, R_2 = n_2 - 3n_1/2+1$ when $n_2\geq 3n_1/2$ and $n_1>1$.
\begin{theorem}
\label{lemma:counterforn} For a directed acyclic graph $G$ with
two $s-t$ pairs, if the connectivity level for $(s_1,t_1)$ is
$n_1$, for $(s_2,t_2)$ is $n_2$, where $n_2\geq 3n_1/2$ and $n_1>1$, there exist instances
that cannot support $R_1=n_1$ and $R_2=n_2-3n_1/2+1$.
\end{theorem}
\begin{proof} Provided in the supplementary documentation.
\end{proof}

\section{Network coding for three unicast sessions - Feasible Instances}
\label{sec:netcod_feas}
It is evident that there exist instances with connectivity level $[2~2~3]$ (and component-wise lower) that are infeasible. Therefore, the possible instances that are potentially feasible are $[1~3~3]$ and $[1~2~4]$, or their permutations and connectivity levels that are greater than them. 
In the discussion below, we show that all the instances with the connectivity levels $[1~3~3]$, $[2~2~4]$ and $[1~2~5]$ are feasible using linear network codes. Our work leaves out one specific connectivity level vector, namely $[1~2~4]$ for which we have been unable to provide either a feasible network code or a network topology where communicating at unit rate is impossible.

As pointed out by the work of \cite{rm}, under linear network coding, the case of multiple unicast requires (a) the transfer matrix for each source-terminal pair to have a rank that is high enough, and (b) the interference at each terminal to be zero. Under random linear network coding, it is possible to assert that the rank of any given transfer matrix from a source $s_i$ to a terminal $t_j$ has w.h.p. a rank equal to the minimum cut between $s_i$ and $t_j$; however, in general this is problematic for satisfying the zero-interference condition. 

Our strategies rely on a combination of graph-theoretic and algebraic methods. Specifically, starting with the connectivity level of the graph, we use graph theoretic ideas to argue that the transfer matrices of the different terminals have certain relationships. The identified relationships then allow us to assert that suitable precoding matrices that allow each terminal to be satisfied can be found. A combination of graph-theoretic and algebraic ideas were also used in the work of \cite{ramamoorthyL12_jnl}, where the problem of multicasting finite field sums over wired networks was considered. However, there are some crucial differences. Reference \cite{ramamoorthyL12_jnl} considered a multicast situation; thus, the issue of dealing with interference did not exist. As will be evident, a large part of the effort in the current work is to demonstrate that the terminals can decode their intended message in the presence of the interfering messages.

We begin with the following definitions.
\begin{definition} {\it Minimality.} Consider a multiple unicast instance $<G = (V,E), \{s_i - t_i\}_{1}^{n}, \{1,~ \dots, ~1\}>$, with connectivity level $[k_1~k_2~\dots~ k_n]$. The graph $G$ is said to be minimal if the removal of any edge from $E$ reduces the connectivity level. If $G$ is minimal, we will also refer to the multiple unicast instance as minimal.
\end{definition}
Clearly, given a non-minimal instance $G = (V,E)$, we can always remove the non-essential edges from it, to obtain the minimal graph $G_{\min}$. This does not affect connectivity. A network code for $G_{\min} = (V, E_{\min})$ can be converted into a network code for $G$ by simply assigning the zero coding vector to the edges in $E \backslash E_{\min}$.
\begin{definition}
{\it Overlap edge.} An edge $e$ is said to be an overlap edge for paths $P_i$ and $P_j$ in $G$, if $e \in P_i \cap P_j$. 
\end{definition}
\begin{definition} {\it Overlap segment.} Consider a set of edges $E_{os} = \{e_1, \dots, e_l\} \subset E$ that forms a path. This path is called an overlap segment for paths $P_i$ and $P_j$ if
\begin{itemize}
\item[(i)] $\forall k \in\{1, \dots, l\}$, $e_k$ is an overlap edge for $P_i$ and $P_j$,
\item[(ii)] none of the incoming edges into tail($e_1$) are overlap edges for $P_i$ and $P_j$, and
\item[(iii)] none of the outgoing edges leaving head($e_l$) are overlap edges for $P_i$ and $P_j$.
\end{itemize}
\end{definition}
Our solution strategy is as follows. We first convert the original instance into another {\it structured} instance where each internal node has at most degree three (in-degree + out-degree). We then convert this new instance into a minimal one, and develop the network code assignment algorithm. This network code, can be converted into a network code for the original instance.

Following \cite{LSB06} we can efficiently construct a {\em structured} graph $\hat{G}=(\hat{V},\hat{E})$ in which each internal node $v \in \hV$ is of total degree at most three with the following properties.
\begin{itemize}
\item[(a)] $\hG$ is acyclic.
\item[(b)] For every source (terminal) in $G$ there is a corresponding source (terminal) in $\hG$.
\item[(c)] For any two edge disjoint paths $P_i$ and $P_j$ for one unicast session in $G$, there exist two {\em vertex} disjoint paths in $\hG$ for the corresponding session in $\hG$.
\item[(d)] Any feasible network coding solution in $\hG$ can be efficiently turned into a feasible network coding solution in $G$.
\end{itemize}
In all the discussions below, we will assume that the graph $G$ is structured. It is clear that this is w.l.o.g. based on the previous arguments.
\subsection{Code Assignment Procedure For Instances With Connectivity Level $[1~3~3]$}
We begin by showing some basic results for two-unicast. The three unicast result follows by applying vector network coding over two time units and using the two-unicast results.

\begin{lemma}
\label{lemma:uneven_two_unicast}
A minimal multiple unicast instance $<G,\{s_1 - t_1, s_2 - t_2\}, \{1, m\}>$ with connectivity level $[1~m+1]$ is always feasible.
\end{lemma}
\begin{proof}
Denote the path from $s_1$ to $t_1$ as $\mathcal P_{1}=\{P_{11}\}$, and the $m+1$ paths from $s_2$ to $t_2$ as $\mathcal P_{2}=\{P_{21},\dots,P_{2m+1}\}$. The information that needs to be transmitted from $s_1$ is $X_1$, and the information that needs to be transmitted from $s_2$ is $X_{21},\dots, X_{2m}$. We assume that $P_{11}$ overlaps with all paths in $\mathcal P_2$. Otherwise, if $P_{11}$ overlaps with $n$ paths in $\mathcal P_2$ where $0\leq n<m+1$, w.l.o.g, assume they are $P_{21},\dots, P_{2n}$. Then $X_{2n},\dots,X_{2m}$ can be simply transmitted over the overlap free paths $P_{2n+1}, \dots, P_{2m+1}$, and the problem reduces to communicating $X_1$ and $X_{21},\dots, X_{2n-1}$ over $P_{11} \cup P_{21} \cup \dots \cup P_{2n}$, which corresponds to the statement of the theorem with $m$ replaced by $n-1$. Hence, we focus on the case that $P_{11}$ overlaps with all paths in $\mathcal P_2$.

We assume that the local coding vectors for each edge are indeterminates for now. Source $s_2$ uses a precoding matrix $\Theta$; the rows of $\Theta$ specify the coding vectors on the outgoing edges of $s_2$. The choice of the local coding vectors and $\Theta$ is discussed below. The transmitted symbol on the outgoing edge from $s_2$ belonging to $P_{2i}$ is $[\theta_{i1}~\cdots~\theta_{im}][X_{21}~\cdots~ X_{2m}]^T$ where $i=1,~\dots~,m+1$. Let $\underline{\theta}_j=[\theta_{1j}~\cdots~\theta_{(m+1)j}]^T$ where $j=1,\dots,m$. 

As $P_{11}$ overlaps with all paths on $\mathcal P_{2}$, there will be many overlap segments on $P_{11}$. Let $E_{os1}$ denote the overlap segment that is closest to $t_1$ (under the topological order imposed by the directed acyclic nature of the graph) along $P_{11}$ and suppose that it is on $P_{21}$. A key observation is that $E_{os1}$ is also the overlap segment on $P_{21}$ that is closest to $t_2$. Indeed if there is another overlap segment $E_{os1}'$ that is closer to $t_2$ along $P_{21}$, then it implies the existence of a cycle in the graph. Let the coding vectors at each intermediate node be specified by indeterminates for now.

The overall transfer matrix from the pair of sources $\{s_1,s_2\}$ to $t_1$ can be expressed as
$$
[M_{11}~|~M_{12}]=[\alpha_1~|~\gamma_{11}~\cdots~\gamma_{1(m+1)}].
$$
Similarly, the transfer matrix from the pair of sources $\{s_1,s_2\}$ to $t_2$ can be expressed as
$$
[M_{21}~|~M_{22}]=\left[{\begin{array}{c|ccc}
\alpha_1&\gamma_{11}&\cdots&\gamma_{1(m+1)}\\
\alpha_2&\gamma_{21}&\cdots&\gamma_{2(m+1)}\\
\vdots&\vdots&\ddots&\vdots\\
\alpha_{m+1}&\gamma_{(m+1)1}&\cdots&\gamma_{(m+1)(m+1)}\end{array}}\right].
$$
The received vector at terminal $t_i$ is therefore $[M_{i1} ~|~ M_{i2}] \left[{\begin{array}{c} X_1 \\ \Theta [X_{21}\cdots X_{2m}]^T \end{array}}\right]$. The variables $\alpha_i's$ and $\gamma_{ij}'s$ in the above matrices depend on the indeterminate local coding vectors and are therefore undetermined at this point.

We emphasize that the first row of $[M_{21}~|~M_{22}]$ is the same as $[M_{11}~|~M_{12}]$. As there exists a single path between $s_1$ and $t_1$, it is clear that $\alpha_1$ is not identically zero. Similarly, as there are $m+1$ edge-disjoint paths between $s_2$ to $t_2$, we have that $\det(M_{22})$ is not identically zero. Now suppose that we employ random linear network coding at all nodes. Using the Schwartz-Zippel lemma \cite{motwaniR}, this implies that $\alpha_1\neq 0$ and $\det(M_{22}) \neq 0$ w.h.p. We assume that $\alpha_1\neq 0$ and $\det(M_{22}) \neq 0$ in the discussion below. Next we select $\theta_{ij}$, $i=1,\dots,m+1$, $j=1,\dots,m$ such that they satisfy the following equation.
\begin{align}
\label{eq:cons133}
M_{22}[\underline{\theta}_1~\cdots~\underline{\theta}_m]&=\left[{\begin{array}{ccc}
0&\cdots&0\\
a_1&\cdots&0\\
\vdots&\ddots&\vdots\\
0&\cdots&a_m\end{array}}\right]
\end{align}
where $a_1,\dots,a_m$ are non-zero values. Note that such $[\underline{\theta}_1~\cdots~\underline{\theta}_m]$ 
can be chosen since $M_{22}$ is full-rank.

Terminal $t_1$ can decode, since $M_{12}[\underline{\theta}_1~\cdots~\underline{\theta}_{m}]=[0\cdots0]$ and $\alpha_1\neq 0$, and $t_2$ can decode, since $X_1$ is available at $t_2$, and $rank(M_{22}[\underline{\theta}_1~\cdots~\underline{\theta}_m])=m$ (from eq. (\ref{eq:cons133})). Finally, we note that there are $q-1$ choices for each $\underline{\theta}_j$.
\end{proof}
We remark that the main issue in the above argument is to demonstrate that the choice of $\Theta$ works simultaneously for both $t_1$ and $t_2$. The observation that $E_{os1}$ is overlap segment closest to $t_1$ and $t_2$ along $P_{11}$ and $P_{21}$ respectively allows us to make this argument.

The result for three unicast sessions with connectivity level $[1~3~3]$ now follows by using vector linear network coding over two time units, as discussed below. 
\begin{theorem}
\label{thm:1_3_3_result}
A multiple unicast instance with three sessions, $<G,\{s_i-t_i\}^3_1,\{1,1,1\}>$ with connectivity level at least $[1 ~ 3 ~3]$ is feasible.
\end{theorem}
\begin{proof}
W.l.o.g. we assume that the connectivity level is exactly $[1~3~3]$. We use vector linear network coding over two time units. For facilitating the presentation we form a new graph $G^*$ where each edge $e \in E$ is replaced by two parallel unit capacity edges $e^1$ and $e^2$ in $G^*$. The messages at source node $s_i$ are denoted $[X_{i1}~X_{i2}], i = 1, \dots, 3$. Let the subgraph of $G^*$ induced by all edges with superscript $i$ be denoted $G^*_i$. In $G^*_1$, there exists a single $s_1 - t_1$ path and three edge disjoint $s_2 - t_2$ paths. Therefore, we can transmit $X_{11}$ from $s_1$ to $t_1$ and $[X_{21}~X_{22}]$ from $s_2$ to $t_2$ using the result of Lemma \ref{lemma:uneven_two_unicast}. Similarly, we use $G_2^*$ to communicate $X_{12}$ from $s_1$ to $t_1$ and $[X_{31}~X_{32}]$ from $s_3$ to $t_3$. Thus, over two time units a rate of $[1~1~1]$ can be supported.
\end{proof}


\vspace{-0.1in}
\subsection{Code Assignment Procedure For Instances With Connectivity Level $[2~2~4]$}
Our solution approach is similar in spirit to the discussion above. In particular, we first investigate a two-unicast scenario with connectivity level $[2~4]$ and rate requirement $\{2,1\}$ and use that in conjunction with vector network coding to address the three-unicast with connectivity level $[2~2~4]$. 
\begin{lemma}
\label{lemma:24}
A minimal multiple unicast instance $<G,\{s_1 - t_1, s_2 - t_2\}, \{2,1\}>$ with connectivity level $[2~4]$ is feasible.
\end{lemma}
\begin{proof}
Let $\mathcal P_1=\{P_{11},P_{12}\}$ denote two edge disjoint paths (also vertex disjoint due to the structured nature of $G$) from $s_1$ to $t_1$ and $\mathcal P_2=\{P_{21}, P_{22}, P_{23}, P_{24}\}$ denote the four vertex disjoint paths from $s_2$ to $t_2$. Let the source messages at $s_1$ be denoted by $X_1$ and $X_2$, and the source message at $s_2$ by $X_3$. We color the edges of the graph such that each edge on $P_{11}$ is colored red, each edge on $P_{12}$ is colored blue and each edge on a path in $\mathcal{P}_2$ is colored black. 

As the paths in $\mathcal{P}_1$ and $\mathcal{P}_2$ are vertex-disjoint, it is clear that a node with an in-degree of two is such that its outgoing edge has two colors (either {\it (blue, black)} or {\it (red, black)}). The path further downstream continues to have two colors until it reaches a node of out-degree two.

Such an overlap segment with two colors will be referred to as a \emph{mixed color overlap segment}. We shall also use the terms \emph{red} or \emph{blue overlap segment} to refer to segments with colors {\it (red, black)} and {\it (blue, black)} respectively. Note that by our naming convention path $P_{ij}$ is a path that enters terminal $t_i$. Under the topological order in $G$ we can identify the overlap segment on $P_{ij}$ that is closest to $t_i$. In the discussion below this will be referred to as the last overlap segment with respect to path $P_{ij}$. Two overlap segments $E_{os1}$ and $E_{os2}$ are said to be neighboring with respect to $P_{ij}$ if there are no overlap segments between them along $P_{ij}$. An example of neighboring overlap segments is shown in Fig. \ref{fig:os}.


\begin{claim}
\label{claim:samepath}
Consider two neighboring mixed color overlap segments $E_{os1}$ and $E_{os2}$ with respect to path $P_{1i} \in \mathcal{P}_1$. Then $E_{os1}$ and $E_{os2}$ cannot lie on the same path $P_{2j} \in \mathcal{P}_2$.
\end{claim}
\emph{proof}:
W.l.o.g., assume that $E_{os1} = \{e_{1}, \dots, e_{k_1}\}$ and $E_{os2} = \{e'_{1}, \dots, e'_{k_2}\}$ are such that $e_{k_1}$ is upstream of $e'_{1}$. Now assume that both $E_{os1}$ and $E_{os2}$ are on $P_{2j}$. Note that $head(e_{k_1})$ has two outgoing edges, one of which belongs to $P_{1i}$ and the other belongs to $P_{2j}$ (denoted by $e^*$). We claim that $e^*$ can be removed while the connectivity level remains the same. This is because $e^*$ does not belong to $P_{1i}$ and $P_{2k}$, $\forall k\neq j$. 
Moreover, after the removal, $P_{2j}$ can be modified to the path specified as $path(s_2, head(e_{k_1})) - path(e_{k_1}, e'_{1}) - path(head(e'_{1}), t_2)$ where $path(e_{k_1}, e'_{k_2})$ is along $P_{1i}$. The new $P_{2j}$ is vertex disjoint of $P_{2k}$, $\forall k\neq j$, since $E_{os1}$ and $E_{os2}$ are neighboring mixed color overlap segments along $P_{1i}$ which means that $path(e_{k_1} - e'_1)$ is either purely blue or purely red. This contradicts the minimality of the graph.  \endproof

Likewise, two neighboring mixed color overlap segments 
with respect to $P_{2i}$, cannot lie on the same path $P_{1j}$. 


To explain our coding scheme, we first denote the last red (blue) overlap segment with respect to $P_{11}$ ($P_{12}$) by $E_{r}$ ($E_{b}$). If there is no $E_{r}$, then $X_1$ can be transmitted along $P_{11}$. According to Lemma \ref{lemma:uneven_two_unicast}, $X_2$ and $X_3$ can be transmitted to $t_1$ and $t_2$ respectively. A similar argument can be applied to the case when there is no $E_{b}$. Hence, we assume that both $E_{r}$ and $E_{b}$ exist. Based on their locations in $G$, we distinguish the following two cases.


\begin{figure}[t]
\center{
\subfigure[]{\label{fig:os}
\includegraphics[width=43mm,clip=false, viewport=60 20 330 350]{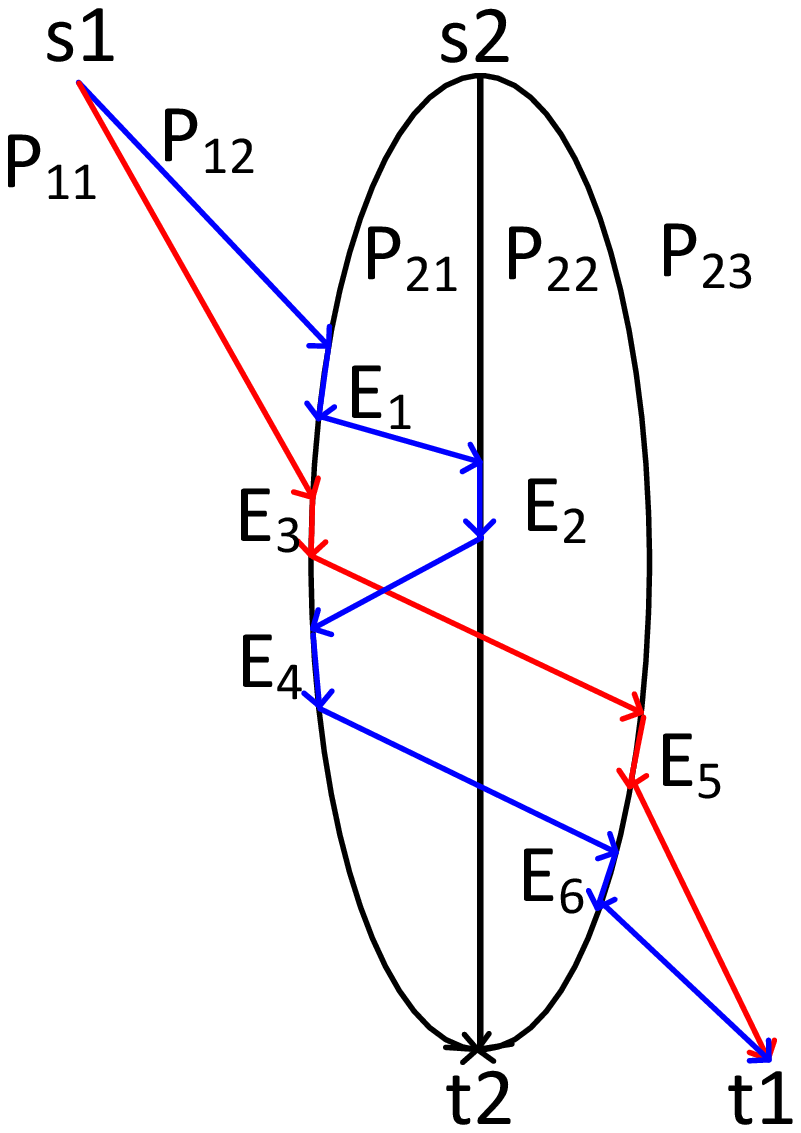}}
\subfigure[]{\label{fig:case24}
\includegraphics[width=42mm,clip=false, viewport=0 40 210 300]{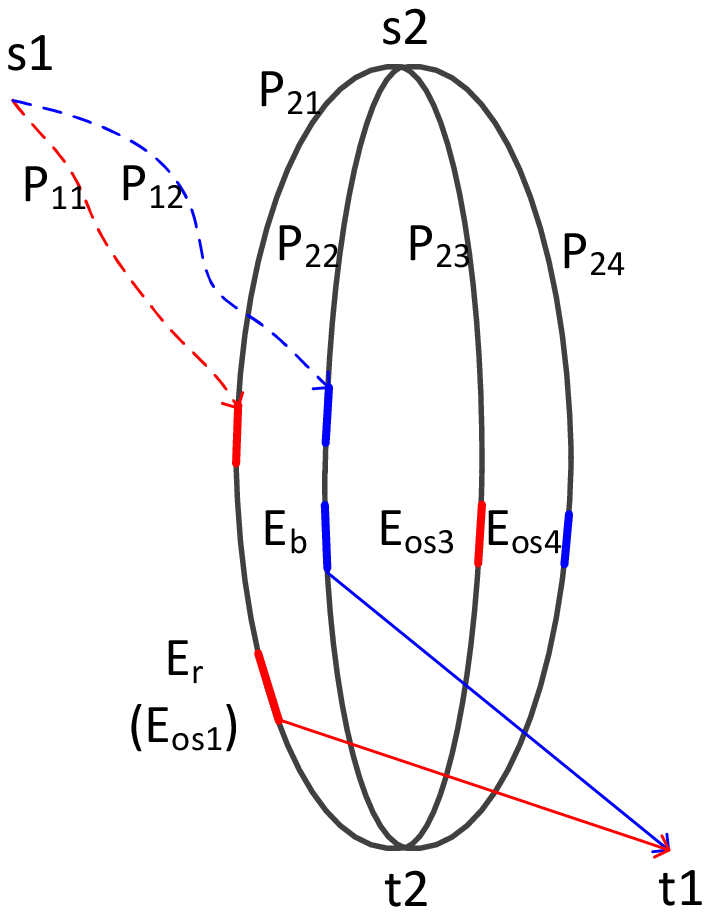}}
\caption{ (a) An instance of network where there are several pairs of neighboring overlap segments. $E_1$ and $E_3$ are neighboring overlap segments along $P_{21}$, $E_1$ and $E_2$ are neighboring overlap segments along $P_{12}$. $E_1$ and $E_4$ are not overlap segments along any paths. (b) A network with connectivity level $[2~4]$ and rate $\{2,1\}$. The coloring of the different paths helps us to show that a linear network coding solution exists.}}
\end{figure}

\begin{list}{}{\leftmargin=0.1in \labelwidth=0cm \labelsep = 0cm}
\item[$\bullet~$] \underline{{\it Case 1: $E_{r}$ and $E_{b}$ are on different paths $\in \mathcal{P}_2$.}}\\ W.l.o.g. we assume that $E_r$ and $E_b$ are on paths $P_{21}$ and $P_{22}$. If there are no mixed color overlap segments on either $P_{23}$ or $P_{24}$, $X_3$ can be transmitted to $t_2$ through the overlap free path, and $X_1, X_2$ can be routed to $t_1$. Therefore, we focus on the case that there are mixed color overlap segments on both $P_{23}$ and $P_{24}$. 
Let $E_{osi}$ denote the last mixed color overlap segments with respect to $P_{2i}$, $i=1,\dots,4$ (see Fig. \ref{fig:case24}).

Our coding scheme is as follows. Symbol $X_i$ is transmitted over the outgoing edge from $s_1$ over $P_{1i}, i =1,2$; symbols $\theta_j X_3$ are transmitted over the outgoing edges of $s_2$ over $P_{2j}$, $j=1,\dots,4$ respectively. The values of $\theta_j \in GF(q)$ will be chosen as part of the code assignment below. Let the coding vectors at each intermediate node be specified by indeterminates for now. The overall transfer matrix from the pair of sources $\{s_1,s_2\}$ to $t_1$ can be expressed as
$$
[M_{11}~|~M_{12}]=\left[{\begin{array}{cc|cccc}
\alpha_1&\beta_1&\gamma_{11}&\gamma_{12}&\gamma_{13}&\gamma_{14}\\
\alpha_2&\beta_2&\gamma_{21}&\gamma_{22}&\gamma_{23}&\gamma_{24}\end{array}}\right],
$$
such that the received vector at $t_1$ is $[M_{11}~|~M_{12}] [X_1 ~X_2 ~|~ \theta_1 X_3~ \dots ~\theta_4 X_3]^T$. Recall that $E_{r}$ and $E_{b}$ are the last mixed color segments with respect to $P_{11}$ and $P_{12}$. Thus, they carry the same information as the incoming edges of $t_1$ which implies that the row vectors of $[M_{11}~|~M_{12}]$ are the coding vectors on $E_{r}$ and $E_{b}$ respectively. Similarly, the transfer matrix from $\{s_1,s_2\}$ to the edge set $\{E_{r},E_{b},E_{os3},E_{os4}\}$ can be expressed as
$$
[M^{e}_{21}~|~M^{e}_{22}]=\left[{\begin{array}{cc|cccc}
\alpha_1&\beta_1&\gamma_{11}&\gamma_{12}&\gamma_{13}&\gamma_{14}\\
\alpha_2&\beta_2&\gamma_{21}&\gamma_{22}&\gamma_{23}&\gamma_{24}\\
\alpha_3&\beta_3&\gamma_{31}&\gamma_{32}&\gamma_{33}&\gamma_{34}\\
\alpha_4&\beta_4&\gamma_{41}&\gamma_{42}&\gamma_{43}&\gamma_{44}\\
\end{array}}\right]
$$
where we use the superscript $e$ to emphasize that these transfer matrices are to the edge set $\{E_{r},E_{b},E_{os3},E_{os4}\}$ and not to the terminal $t_2$.

Note that the entries of the transfer matrices above are functions of the choice of the local coding vectors in the network which are indeterminate. Thus, at this point, the $M_{ij}$ and $M^e_{ij}$ matrices are also composed of indeterminates.

As there exist two edge disjoint paths from $s_1$ to $\{E_{r},E_{b}\}$, the determinant of $M_{11}$
is not identically zero. Similarly, since the edges $E_{r}$, $E_{b}$, $E_{os3}$ and $E_{os4}$ lie on different paths in $\mathcal{P}_{2}$, there are four edge disjoint paths from $s_2$ to the edge subset $\{E_{r},E_{b},E_{os3},E_{os4}\}$, and the determinant of $M^e_{22}$
is not identically zero. 
This implies that their product is not identically zero. Hence, by the Schwartz-Zippel
lemma \cite{motwaniR}, under random linear network coding there exists an assignment of
local coding vectors so that  $rank(M_{11})=2$ and $rank(M^e_{22})=4$. We assume that the local coding vectors are chosen from a large enough field $GF(q)$ so that this is the case. For this choice of local coding vectors we propose a choice of $\underline{\theta} = [\theta_1~\theta_2~\theta_3~\theta_4]^T$ such that the decoding is simultaneously successful at both $t_1$ and $t_2$.

\noindent \underline{Decoding at $t_1$:} As $M_{11}$ is a square full-rank matrix, we only need to null the interference from $s_2$. Accordingly, we choose $\underline{\theta}$ from the null space of $M_{12}$, i.e.,
\vspace{-0.05in}
\begin{equation} M_{12} \underline{\theta} = 0. \end{equation} There are at least $q^2 - 1 $ such non-zero choices for $\underline{\theta}$ as $M_{12}$ is a $2 \times 4$ matrix.


\noindent \underline{Decoding at $t_2$:} The primary issue is that one needs to demonstrate that the choice of $\underline{\theta}$ allows both terminals to simultaneously decode. Indeed, it may be possible that our choice of $\underline{\theta}$ along with a specific network topology may make it impossible to decode at $t_2$. The key argument that this does not happen requires us to leverage certain topological properties of the overlap segments, that we present below.
\begin{claim}
In $G$ either one or both of the following statements hold. (i) $E_{r}$ is the last overlap segment w.r.t. $P_{21}$. (ii) $E_b$ is the last overlap segment w.r.t. $P_{22}$.
\end{claim}
\begin{proof}
Assume that neither statement is true. This means that there is a blue overlap segment $E'_{b}$ below $E_{r}$ along $P_{21}$, and there is a red overlap segment $E'_{r}$ below $E_{b}$ along $P_{22}$. Thus, $E'_{r}$ is upstream of $E_{r}$ and $E'_{b}$ is upstream of $E_{b}$. However, this means that edges $E'_{r}$, $E_{r}$, $E'_{b}$ and $E_{b}$ form a cycle, which is a contradiction.
\end{proof}
In the discussion below, w.l.o.g., we assume that $E_{r}$ is the last overlap segment on $P_{21}$. The argument above allows us to identify edges $E_{r}$, $E_{os3}$ and $E_{os4}$ that carry the {\it same symbols} as those entering $t_2$.
We show below that the $X_1$ and $X_2$ components can be canceled by using the information on $E_{os3}$ and $E_{os4}$ while retaining the $X_3$ component.

Let $\underline{\gamma}_i$ represent the vector $[\gamma_{i1}~\gamma_{i2}~\gamma_{i3}~\gamma_{i4}]^T, i = 1, \dots, 4$ in the discussion below. Note that if $[\alpha_3~\beta_3]$ and $[\alpha_4~\beta_4]$ are linearly independent, there exist $\delta_3$ and $\delta_4$ such that
\begin{equation*} [\alpha_1~\beta_1]=\delta_3[\alpha_3~\beta_3]+\delta_4[\alpha_4~\beta_4], \end{equation*}
where $\delta_3$ and $\delta_4$ are not both zero. Thus, $t_2$ can recover $[-\underline{\gamma}_1 + \delta_3 \underline{\gamma}_3 + \delta_4 \underline{\gamma}_4]^T \underline{\theta} X_3$. Note that $\underline{\gamma}_1^T \underline{\theta} = 0$, by the constraint on $\underline{\theta}$ above, thus we only need to pick $\underline{\theta}$ such that $[\delta_3\underline{\gamma}_3 + \delta_4 \underline{\gamma}_4]^T \underline{\theta} \neq 0$. To see that this can be done, we note that $M_{22}$ is full rank which implies that the matrix 
$[\underline{\gamma}_1 ~~ \underline{\gamma}_2 ~~ (\delta_3\underline{\gamma}_3 + \delta_4 \underline{\gamma}_4)]^T$ is full rank. Therefore, there exist at most $q$ choices for $\underline{\theta}$ such that $[\underline{\gamma}_1 ~~ \underline{\gamma}_2 ~~ (\delta_3\underline{\gamma}_3 + \delta_4 \underline{\gamma}_4)]^T \underline{\theta} = 0$. Hence, there are at least $q^2 - q - 1 > 0$ non-zero choices for $\underline{\theta}$ that allow decoding at $t_1$ and $t_2$ simultaneously.

If $[\alpha_3~\beta_3]$ and $[\alpha_4~\beta_4]$ are dependent, decoding can be performed simply by working only with the received values over $E_{os3}$ and $E_{os4}$ using a similar argument as above.

\item[$\bullet~$] \underline{{\it Case 2: $E_{r}$ and $E_{b}$ are on the same path $P_{2i}$.}}\\ W.l.o.g., assume that $E_{b}$ is downstream of $E_{r}$ along $P_{21}$. Then $E_b$ will be the last overlap segment w.r.t. $P_{21}$. Let $E'_{b}$ denote the blue overlap segment that is a neighbor of $E_{b}$ w.r.t. $P_{12}$. Note that $E'_{b}$ cannot be on $P_{21}$ according to Claim \ref{claim:samepath}. If $E'_{b}$ does not exist, it implies that there is only one blue overlap segment (namely, $E_b$) in the network. Therefore, there only exist red overlap segments on $P_{23}$ and $P_{24}$; using Lemma \ref{lemma:uneven_two_unicast}, $X_1$ and $X_3$ can be transmitted to $t_1$ and $t_2$ respectively over $P_{11} \cup P_{23} \cup P_{24}$, and $X_2$ can be routed along $P_{12}$ to $t_1$.

    We now focus on the case when an $E'_{b}$ exists and assume (w.l.o.g.) that it is on $P_{22}$. The main difference is that instead of using random coding over the entire graph, we modify our coding scheme such that random coding is performed over the graph except at $E_b$ and all the edges downstream of $E_b$. At $E_b$, deterministic coding is performed such that $E_b$ carries the same information as the incoming edge of it along $P_{12}$. The information on $E_b$ is further routed to all the downstream edges of $E_b$. Note that by the deterministic coding, $E_b$ carries the same information as $E'_b$.

\noindent \underline{Decoding at $t_1$:}
Using the arguments developed in Case 1, it is clear that $X_1$ and $X_2$ can be decoded from the information on $E'_{b}$ and $E_{r}$. The code assignment ensures that $E_b$ and $E'_b$ carry the same information, thus $t_1$ is satisfied.

\noindent \underline{Decoding at $t_2$:}
In Case 1, we showed that $X_3$ can be decoded from the information on $E_{r}$, $E_{os3}$ and $E_{os4}$. A similar argument can be made that $X_3$ can be decoded from the information on $E'_{b}$, $E_{os3}$ and $E_{os4}$. Since $E_{b}$ carries the same information as $E'_{b}$ and $E_b$ is the last overlap segment on $P_{21}$, terminal $t_2$ can decode $X_3$ by the information on $E_{b}$, $E_{os3}$ and $E_{os4}$.
\end{list}
\end{proof}

By using the result of Lemma \ref{lemma:24} and the idea of vector network coding, we have the following theorem when the connectivity level is $[2~2~4]$.

\begin{theorem}
A multiple unicast instance with three sessions, $<G,\{s_i-t_i\}^3_1,\{1,1,1\}>$ with connectivity level at least $[2~2~4]$ is feasible.
\end{theorem}
\begin{proof} It can be seen that the line of argument used in the proof of Theorem \ref{thm:1_3_3_result}, namely using vector network coding over two time units and use the result of Lemma \ref{lemma:24} gives us the desired result.
\end{proof}

\subsection{Code Assignment Procedure For Instances With Connectivity Level $[1~2~5]$}
We now consider network code assignment for networks where the connectivity level is $[1~2~5]$. The code assignment in this case requires somewhat different techniques. In particular, the idea of using a two-session unicast result along with vector network coding does not work unlike the cases considered previously. At the top level, we still use random network coding followed by appropriate precoding to align the interference seen by the terminals. However, as we shall see below, we will need to depart from a purely random linear code in the network in certain situations. 

As before, we consider a minimal structured graph $G$ and let $X_i$ be the source symbol at source node $s_i$ for $i=1, \dots, 3$ and $\mathcal P_1=\{P_{11}\}$ denote the path from $s_1$ to $t_1$, $\mathcal P_2=\{P_{21},P_{22}\}$ denote the edge disjoint paths from $s_2$ to $t_2$, $\mathcal P_3=\{P_{31}, P_{32}, P_{33}, P_{34}, P_{35}\}$ denote the edge disjoint paths from $s_3$ to $t_3$.



Our scheme operates as follows: $X_1$ is transmitted over the outgoing edge from $s_1$ along $P_{11}$ , $\xi_i X_2$ are transmitted over the outgoing edges of $s_2$ along $P_{2i}$, $i=1,2$, and $\theta_{j}X_3$ are transmitted over the outgoing edges of $s_3$ along  $P_{3j},~j=1,\dots,5$ where $\underline{\xi}= [\xi_1~\xi_2]^T$ and $\underline{\theta} = [\theta_1~\dots~\theta_5]^T$ are precoding vectors chosen from a finite field with size $q$. 


Let $M_i = [M_{i1} ~|~ M_{i2} ~|~ M_{i3}]$ denote the transfer matrix from $\{s_1,s_2,s_3\}$ to terminal $t_i$. Each $M_{ij}$ corresponds to the transformation from source $s_j$ to terminal $t_i$, i.e., the number of columns in $M_{ij}$ is $1, 2$ and $5$ for $j=1, 2$ and $3$ respectively. Similarly, the number of rows in $M_{ij}$ is $1, 2$ and $5$ for $i =1,2$ and $3$ respectively.

In the discussion below we will need to refer to the individual entries of $M_1$ and $M_2$. Accordingly, we express these matrices explicitly as follows.
\begin{align*}
M_1 &= [M_{11} ~|~ M_{12} ~|~ M_{13}] = \left[\alpha_1~|~\underline{\beta}^T~|~\underline{\gamma}^T~\right]\\ &=\left[\alpha_1~|~\beta_1~\beta_2~|~\gamma_1~\gamma_2~\gamma_3~\gamma_4~\gamma_5\right],\\
M_2&=[M_{21}~|~M_{22}~|~M_{23}]= \left[{\begin{array}{c|c|c}
\alpha_1'&\underline{\beta'}^T_1& \underline{\gamma'}^T_1\\
\alpha_2'&\underline{\beta'}^T_2& \underline{\gamma'}^T_2\\
\end{array}}\right]\\
 &=\left[{\begin{array}{c|cc|ccccc}
\alpha'_1&\beta'_{11}&\beta'_{12}&\gamma'_{11}&\gamma'_{12}&\gamma'_{13}&\gamma'_{14}&\gamma'_{15}\\
\alpha'_2&\beta'_{21}&\beta'_{22}&\gamma'_{21}&\gamma'_{22}&\gamma'_{23}&\gamma'_{24}&\gamma'_{25}\\
\end{array}}\right],
\end{align*}
where the entries of the matrices above are functions of indeterminate local coding vectors. The cut conditions imply that $\det(M_{ii})$ is not identically zero for $i = 1, \dots, 3$, and furthermore that their product $\det(M_{11})\det(M_{22})\det(M_{33})$ is not identically zero.




Our solution proceeds as follows.
We first identify a minimal structured subgraph $G'$ of $G$ with the following properties. 
\begin{itemize}
\item[(i)] There exists a path $P'_{11}$, from $s_1$ to $t_1$, \item[(ii)] vertex disjoint paths $P'_{21}$ and $P'_{22}$ from $s_2$ to $t_2$, \item[(iii)] path $\pm$ from $s_1$ to $t_2$ and \item[(iv)] path $\pn$ from $s_2$ to $t_1$.
\end{itemize}
Again, $G'$ is said to be minimal if the removal of any edge from it causes one of the above properties to fail. We note that it is possible that there do not exist any paths from $s_1$ to $t_2$ or from $s_2$ to $t_1$ in $G$. These situations are considered below.

Our analysis depends on the following topological properties of $G'$.\\
\noindent \underline{{\it Case 1:}} The graph $G'$ is such that
\begin{itemize}
\item there is no path from $s_1$ to $t_2$ in $G'$, i.e., $\pm=\emptyset$ (this happens only if there is no path from $s_1$ to $t_2$ in $G$), or
\item there is no path from $s_2$ to $t_1$ in $G'$, i.e., $\pn=\emptyset$ (this happens only if there is no path from $s_2$ to $t_1$ in $G$), or
\item there are paths $\pm$ and $\pn$ in $G'$, and there are overlap segments between $P'_{11}$ and $P'_{21}\cup P'_{22}$.
\end{itemize}

\noindent \underline{{\it Case 2:}} The graph $G'$ is such that
\begin{itemize}
\item there are paths $\pm$ and $\pn$ in $G'$, and $P'_{11}$ does not overlap with either $P'_{21}$ or $P'_{22}$.
\end{itemize}
We emphasize that together Case 1 and Case 2 cover all the possible types of subgraphs for $G'$. Specifically, either $\pm = \emptyset$ or $\pn = \emptyset$. If both $\pm$ and $\pn$ exist in $G'$, then either there are overlaps between $P'_{11}$ and $P'_{21}\cup P'_{22}$ or there are not.
\begin{theorem}
\label{th:case1}
A multiple unicast instance with three sessions, $<G,\{s_i-t_i\}^3_1,\{1,1,1\}>$, with connectivity level $[1~2~5]$ is feasible. 
\end{theorem}
\begin{figure}[htbp]
\hspace{-0.1in}\center{
\subfigure[]{\label{fig:t2ex2}
\includegraphics[width=42mm,clip=false, viewport=0 0 270 340]{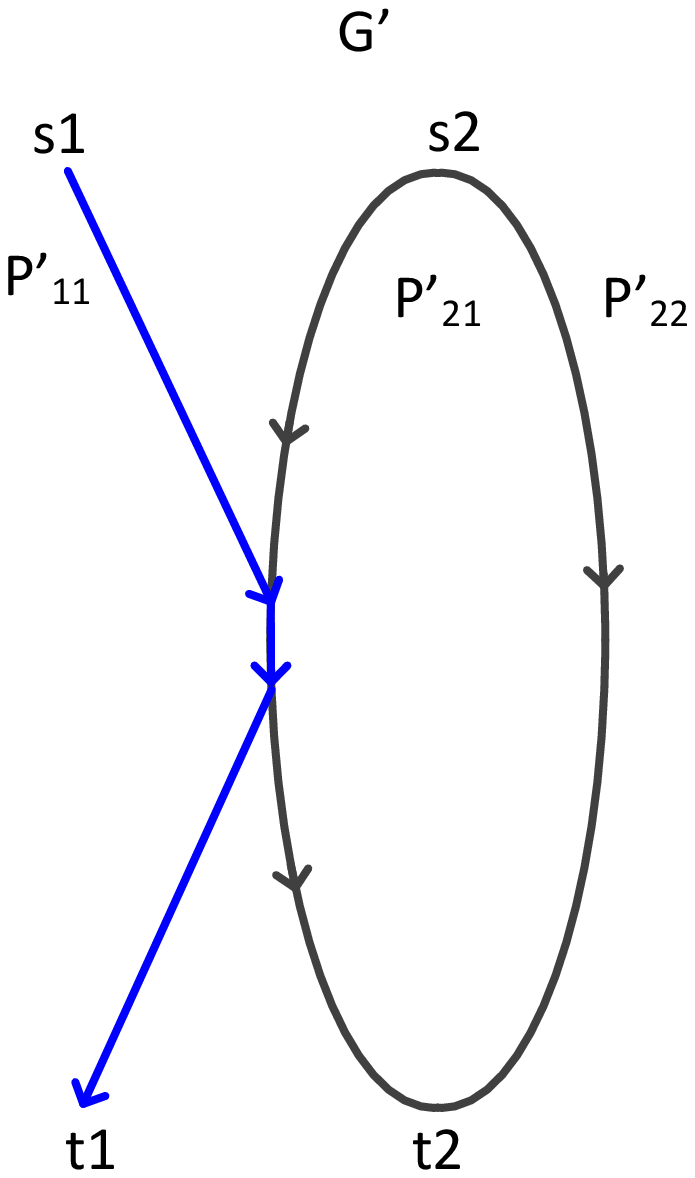}\vspace{-2mm}}
\subfigure[]{\label{fig:t2ex3}
\includegraphics[width=42mm,clip=false, viewport=0 0 270 340]{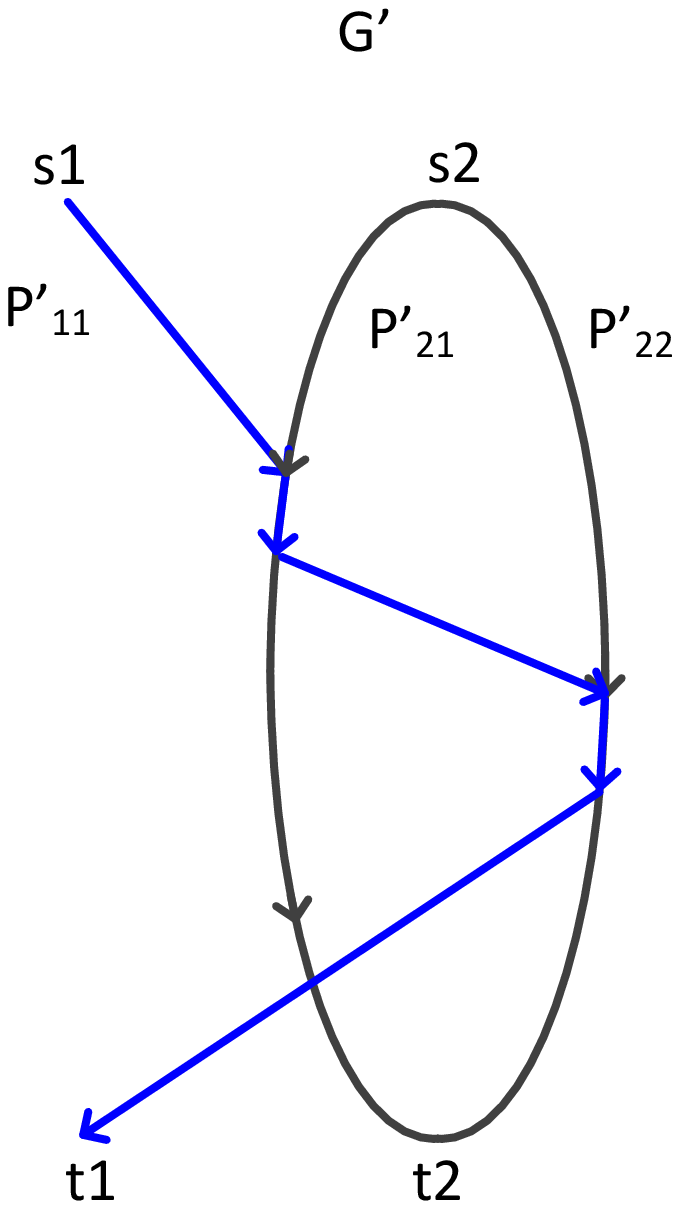}\vspace{-2mm}}
%
\caption{(a) Subgraph $G'$ when $P'_{11}$ overlap with $P'_{21}$. (b) Subgraph $G'$ when $P'_{11}$ overlap with both $P'_{21}$ and $P'_{22}$.}}
\end{figure}

\begin{proof}
We break up the proof into two parts based on type of the subgraph $G'$ that we can find in $G$.

\noindent \uline{Proof when there exists a subgraph $G'$ that satisfies the conditions of Case 1}\\
We perform random linear coding over the graph $G$ over a large enough field. In the discussion below, we will leverage the fact that multivariate polynomials that are not identically zero, evaluate to a non-zero value w.h.p. under a uniformly random choice of the variables. This is needed at several places. By using standard union bound techniques, we can claim that our strategy works w.h.p.

In particular, in the discussion below, we assume that the matrices $M_{ii}, i = 1, \dots, 3$ are full rank and design appropriate precoding vectors $\underline{\xi}$ and $\underline{\theta}$.

\noindent \underline{Decoding at $t_1$:}
For $t_1$ to decode $X_1$, we need to have $\alpha_1 \neq 0$ and the precoding constraints
\begin{equation}
\label{eq:cons11}
[\beta_1~\beta_2]\underline{\xi}=0, \text{~and}
\end{equation}
\begin{equation}
\label{eq:cons12}
[\gamma_1~\gamma_2~\gamma_3~\gamma_4~\gamma_5]\underline{\theta}=0.
\end{equation}

There are at least $q-1$ non-zero vectors $\underline{\xi}$ and $q^4-1$ non-zero vectors $\underline{\theta}$ that can be selected from the field of size $q$ such that eq. (\ref{eq:cons11}) and eq. (\ref{eq:cons12}) are satisfied.

\noindent \underline{Decoding at $t_2$:}

We begin by noting that since $rank(M_{22}) = 2$, $M_{22}\underline{\xi}\neq0$, as long as $\xi \neq 0$.
Next, we argue according to the topological structure of $G'$. The following possibilities can occur.
\begin{list}{}{\leftmargin=0.1in \labelwidth=0cm \labelsep = 0cm}
\item[$(i)$~] {\it There is no path from $s_1$ to $t_2$ in $G'$, i.e., $\pm=\emptyset$.} This implies that $\alpha'_1=\alpha'_2=0$ and in $G$, interference at $t_2$ only exists from $s_3$. Next, at least one component of $M_{22}\underline{\xi}$ will be non-zero, based on the argument above; w.l.o.g. assume that it is the first component. We choose $\underline{\theta}$ to satisfy
    \begin{equation}
    \label{eq:case1t2}
    \underline{\gamma'}_1^T \underline{\theta} = 0.
    \end{equation}
    It is evident that there are at least $q^3 -1 $ non-zero choices of $\underline{\theta}$ that satisfy the required constraints on $\underline{\theta}$ (eqs. (\ref{eq:cons12}) and (\ref{eq:case1t2})). Hence $t_2$ can decode.

\item[$(ii)$~] {\it There exists a path $\pm$ from $s_1$ to $t_2$, i.e., $\pm \neq \emptyset$.}. This means that $M_{21}$ is not identically zero. Here, we first align the interference from $s_3$ within the span of interference from $s_1$ by selecting an appropriate $\underline{\theta}$. We have the following lemma.

\begin{lemma}
If $M_{21}\neq 0$, there exist at least $q^4-1$ choices for $\underline{\theta}$ such that
\vspace{-0.05in}
\begin{equation}
\label{eq:cons22}
M_{23}\underline{\theta}=cM_{21}
\end{equation}
where $c$ is some constant.
\end{lemma}
\begin{proof}
First, w.l.o.g., we assume $\alpha'_2\neq 0$. Hence, there exists a full rank $2\times 2$ upper triangular matrix $U$ such that $U M_{21} = [ 0 ~\alpha'_2]^T$.
Next, define
\begin{equation}
\label{eq:cons23}
[1~0]U M_{23} = \widetilde{\underline{\gamma}}_1^{'T}
\end{equation}
and choose $\underline{\theta}$ to satisfy $\widetilde{\underline{\gamma}}_1^{'T} \underline{\theta} = 0$ and set $c = \underline{\gamma}_2^{'T} \underline{\theta}/\alpha'_2$. Upon inspection, it can be verified that this implies that $U M_{23} \underline{\theta} = c U M_{21}$. As $U$ is invertible, and there is only one linear constraint on $\underline{\theta}$, we have the required conclusion.   
\end{proof}
Thus, under this choice of $\underline{\theta}$, the interference from $s_3$ is aligned within the span of the interference from $s_1$ at $t_2$. Let $\underline{X} = [ X_1 ~X_2~X_3]^T$. The received signal at $t_2$ is
\begin{equation}
\label{eq:re_t2}
[M_{21}~M_{22}\underline{\xi}~M_{23}\underline{\theta}] \underline{X}
=[M_{21}~M_{22}\underline{\xi}]\left[{\begin{array}{c}
X_1+cX_3\\
X_2
\end{array}}\right].
\end{equation}

The following claim concludes the decoding argument for $t_2$.

\begin{claim}
\label{claim:fullrank}
If $M_{21}$ is not identically zero, under random linear coding w.h.p., there exists a $\underline{\xi}$ such that $rank[M_{21}~~M_{22}\underline{\xi}]=2$ and $[\beta_1~\beta_2] \underline{\xi} = 0$.
\end{claim}

\begin{proof} 
We will show that there exists an assignment of local coding vectors such that $\det[M_{21}~~M_{22}\underline{\xi}]\neq0$. This will imply that w.h.p. under random linear coding, this property continues to hold.

Suppose that there is no path from $s_2$ to $t_1$ in $G$, i.e., $\pn=\emptyset$ and $[\beta_1~ \beta_2]$ is identically zero. This does not impose any constraint on $\underline{\xi}$. Next, $M_{22}$ is full rank w.h.p. Hence, we can choose a $\underline{\xi}$ such that required condition is satisfied.

If there exists a path $\pn$ from $s_2$ to $t_1$ in $G'$, $[\beta_1~~\beta_2]$ is not identically zero. W.l.o.g., we assume that $\beta_1$ is not identically zero. By Lemma \ref{lemma:common} (see Appendix), proving that $\det[M_{21}~~M_{22}\underline{\xi}]\neq0$, is equivalent to checking that the determinant in (\ref{eq:deterin}) is not identically zero. Now we demonstrate that there exists a set of local coding vectors such that the determinant in (\ref{eq:deterin}) is non-zero. We consider the subgraph $G'=P'_{11}\cup P'_{21}\cup P'_{22}\cup \pm\cup\pn$ (identified above) - our choice of the coding vectors on all the other edges will be assigned to the zero vector. As both $\pm\neq \emptyset$ and $\pn\neq \emptyset$, we only consider the case where $P'_{11}$ overlaps with $P'_{21}\cup P'_{22}$. We distinguish the following cases.
\begin{enumerate}
\item {\it $P'_{11}$ overlaps with either $P'_{21}$ or $P'_{22}$.} W.l.o.g., assume it is $P'_{21}$. First note that when $P'_{11}$ overlap with one of $P'_{21}$ and $P'_{22}$ in $G'$, there is a path from $s_1$ to $t_2$ and a path from $s_2$ to $t_1$ in $P'_{11}\cup P'_{21}\cup P'_{22}$. Hence, $G'$ can be completely represented by $P'_{11}\cup P'_{21}\cup P'_{22}$. This is shown in Fig. \ref{fig:t2ex2}. It is evident that we can choose coding coefficients such that
$$[\beta_1~~\beta_2]=[1~0], \text{~and}$$
\begin{equation}
[M_{21}~~M_{22}]=\left[{\begin{array}{ccc}
1&1&0\\
0&0&1\\
\end{array}}\right].
\end{equation}

By substituting them into eq. (\ref{eq:deterin}), the determinant of $[M_{21}~~M_{22}\underline{\xi}]$ is not zero. 

\item {\it $P'_{11}$ overlaps with both $P'_{21}$ and $P'_{22}$.} Using a similar argument as above, $G'$ can be completely represented by $P'_{11}\cup P'_{21}\cup P'_{22}$ if $P'_{11}$ overlaps with both $P'_{21}$ and $P'_{22}$. 
    Note that there will be one overlap between $P'_{11}$ and each of $P'_{21}$ and $P'_{22}$. Otherwise, assume there are two overlaps between $P'_{11}$ and $P'_{21}$, then some edges can be removed without contradicting the minimality of the graph $G'$. This is shown in Fig. \ref{fig:t2ex3}. Assume $P'_{11}$ overlap with $P'_{21}$ first. We can find a set of coding coefficients such that  
$$[\beta_1~~\beta_2]=[1~1] \text{~and}$$
\begin{equation}
[M_{21}~~M_{22}]=\left[{\begin{array}{ccc}
1&1&0\\
1&1&1\\
\end{array}}\right].
\end{equation}

By substituting them into eq. (\ref{eq:deterin}), the determinant of $[M_{21}~~M_{22}\underline{\xi}]$ is not zero. 
\end{enumerate}
In both cases, therefore the required condition holds w.h.p. under random linear coding.
\end{proof}
\noindent Terminal $t_2$ can decode since it can solve the system of equations specified by in eq. (\ref{eq:re_t2}).

\end{list}

\noindent \underline{Decoding at $t_3$:}
At $t_3$, we need to decode $X_3$ in the presence of the interference from $s_1$ and $s_2$. The prior constraints on $\underline{\theta}$, namely (\ref{eq:cons12}) and (\ref{eq:case1t2}) for case (i), or (\ref{eq:cons12}) and (\ref{eq:cons22}) for case (ii) allow at least $q^3 - 1$ choices for it. As $M_{33}$ is full-rank, this implies that there are at least $q^3 - 1$ corresponding distinct $M_{33} \underline{\theta}$ vectors.  
Next, for $t_3$ to decode $X_3$, from Lemma \ref{lemma:partialDecode}, we need to have
\begin{equation}
\label{eq:cons31}
M_{33}\underline{\theta}\notin span([M_{31}~~M_{32}\underline{\xi}]).
\end{equation}
Since there are at most $q^2$ vectors in $span([M_{31}~~M_{32}\underline{\xi}])$, there are at least $q^3-q^2-1 > 0$ choices for $\underline{\theta}$ such that all the required constraints on $\underline{\theta}$ are satisfied.

%
\noindent \uline{Proof when there exists a subgraph $G'$ that satisfies the conditions of Case 2}\\
As before, our overall strategy will be to use random linear network coding, however in certain cases we will need to make modifications to the code assignment. We argue based on the properties of the minimal structured subgraph $G'$. Recall that under Case 2, paths $\pm$ and $\pn$ exist and $P'_{11}$ does not overlap with $P'_{21}\cup P'_{22}$. As the graph is structured, this implies that $P'_{11}$, $P'_{21}$ and $P'_{22}$ are all vertex disjoint.

Our first goal is to show that $G'$ is topologically equivalent to one of the graphs shown in Figs. \ref{fig:t2ex11}, \ref{fig:t2ex12} and \ref{fig:t2ex13}. Towards this end, we color $P'_{11}\cup P'_{21}\cup P'_{22}$ black, the path $\pm$ red, and the path $\pn$ blue. In this process, certain edges will get a set of colors (which are a subset of $\{red, blue, black\}$). Note that there cannot be any edge that has the color $\{blue, red\}$. To see this, assume otherwise: then one could find a new path from $s_1$ to $t_1$ that overlaps $\pm$ and $\pn$ and delete at least one edge from $P'_{11}$, contradicting the minimality of $G'$. By similar arguments, $\pm$ and $\pn$ cannot overlap on $P'_{21}\cup P'_{22}$. Hence,  paths $\pm$ and $\pn$ can only overlap if they also overlap with $P'_{11}$.


Next, we identify certain special edges in $G'$. As there is only one path going out of $s_1$, $P'_{11}$ and $\pm$ will overlap. A similar argument shows that $P'_{11}$ and $\pn$ will overlap. Likewise, $\pm$ and $\pn$ will overlap with $P'_{21}$ or $P'_{22}$. Consider, the overlap between $P'_{11}$ and $\pm$. Using the minimality of $G'$ it can be seen that there can be exactly one overlap segment between them; we identify the edge  $\in P'_{11} \cap \pm$ at the farthest distance from $s_1$, such that it has two outgoing edges belonging to exclusively $P'_{11}$ and $\pm$, and call it $e_1$. Similarly, we identify the edge $\in P'_{11}\cap \pn$ that is closest to $s_1$, and call it $e_3$.

Next, consider the overlap between $\pm$ and $P'_{21}\cup P'_{22}$. Once again, by minimality it holds that there is exactly one contiguous overlap segment between $\pm$ and $P'_{21}\cup P'_{22}$, that can either be on $P'_{21}$ or $P'_{22}$. We identify $e_4$ as the edge in $\pm \cap (P'_{21}\cup P'_{22})$ that is closest to $s_1$. In a similar manner, $e_2$ is identified as the edge $\pn \cap (P'_{21}\cup P'_{22})$ that is farthest away from $s_2$.


We now consider the possible orders of the edges $e_1, \dots, e_4$. As $e_1$ and $e_3$ belong to $P'_{11}$, one of them has to be downstream of the other along $P'_{11}$. Consider the following cases.
\begin{itemize}
\item {\it $e_3$ is downstream of $e_1$ along $P'_{11}$.} If edges $e_2$ and $e_4$ lie on the same path $\in \{P'_{21}, P'_{22}\}$, we first note that $e_4$ has to be downstream of $e_2$ (by minimality, otherwise the segment between $e_1$ and $e_3$ along $P'_{11}$ can be removed); the graph $G'$ is topographically equivalent to Fig. \ref{fig:t2ex11}. If $e_2$ and $e_4$ lie on different paths $\in \{P'_{21}, P'_{22}\}$, the graph $G'$ is topographically equivalent to Fig. \ref{fig:t2ex12}.
\item {\it $e_1$ is downstream of $e_3$ along $P'_{11}$, or $e_1 = e_3$.} In this case $e_2$ and $e_4$ have to lie on different paths $\in \{P'_{21}, P'_{22}\}$. To see this, assume they both lie on $P'_{21}$: if $e_4$ is downstream of $e_2$, the minimality of $G'$ does not hold (segment between $e_2$ and $e_4$ along $P'_{21}$ can be removed), whereas if $e_2$ is downstream of $e_4$, the acyclicity of $G'$ is contradicted. Therefore, the only possibility is that $e_2$ and $e_4$ lie on different paths  $\in \{P'_{21}, P'_{22}\}$ and in this case $G'$ is topographically equivalent to Fig. \ref{fig:t2ex13}.
\end{itemize}
With the above arguments in place, it is clear that $G'$ is topographically equivalent to one of the graphs in Fig. \ref{fig:t2ex11}, \ref{fig:t2ex12} or \ref{fig:t2ex13}.

\begin{figure*}[htbp]
\hspace{-0.1in}\center{
\subfigure[]{\label{fig:t2ex11}
\includegraphics[width=45mm,clip=false, viewport=0 10 270 300]{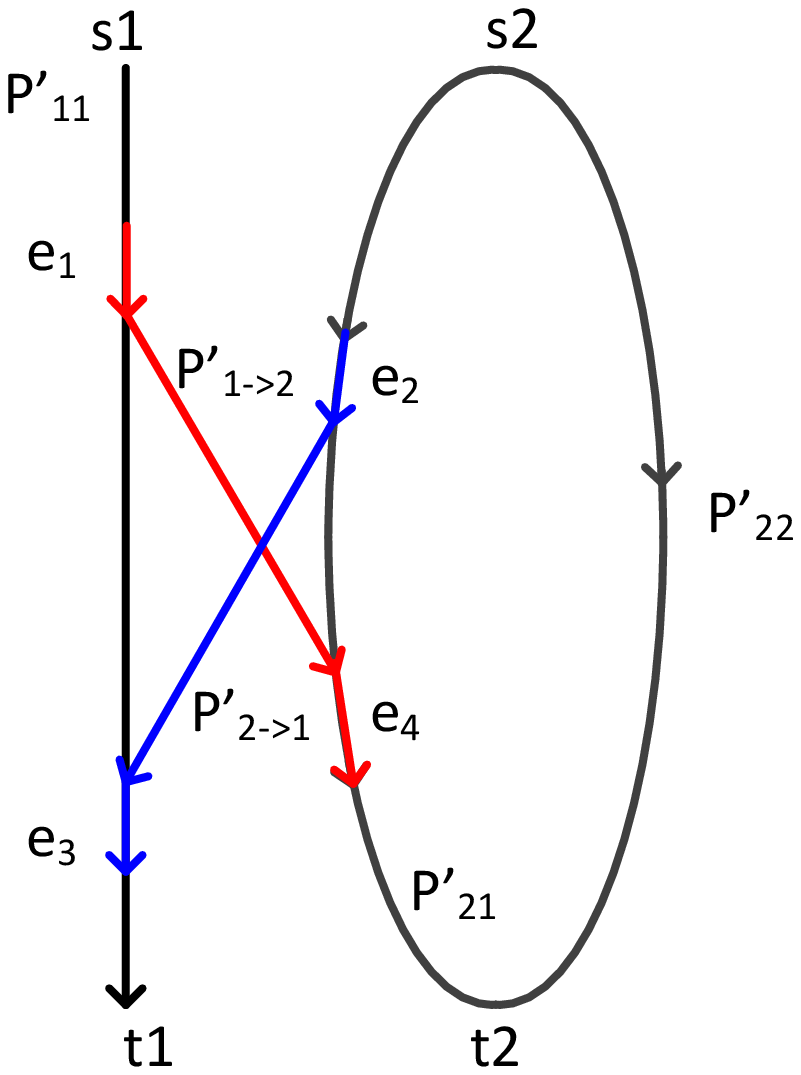}\vspace{-2mm}}
\subfigure[]{\label{fig:t2ex12}
\includegraphics[width=45mm,clip=false, viewport=0 10 270 300]{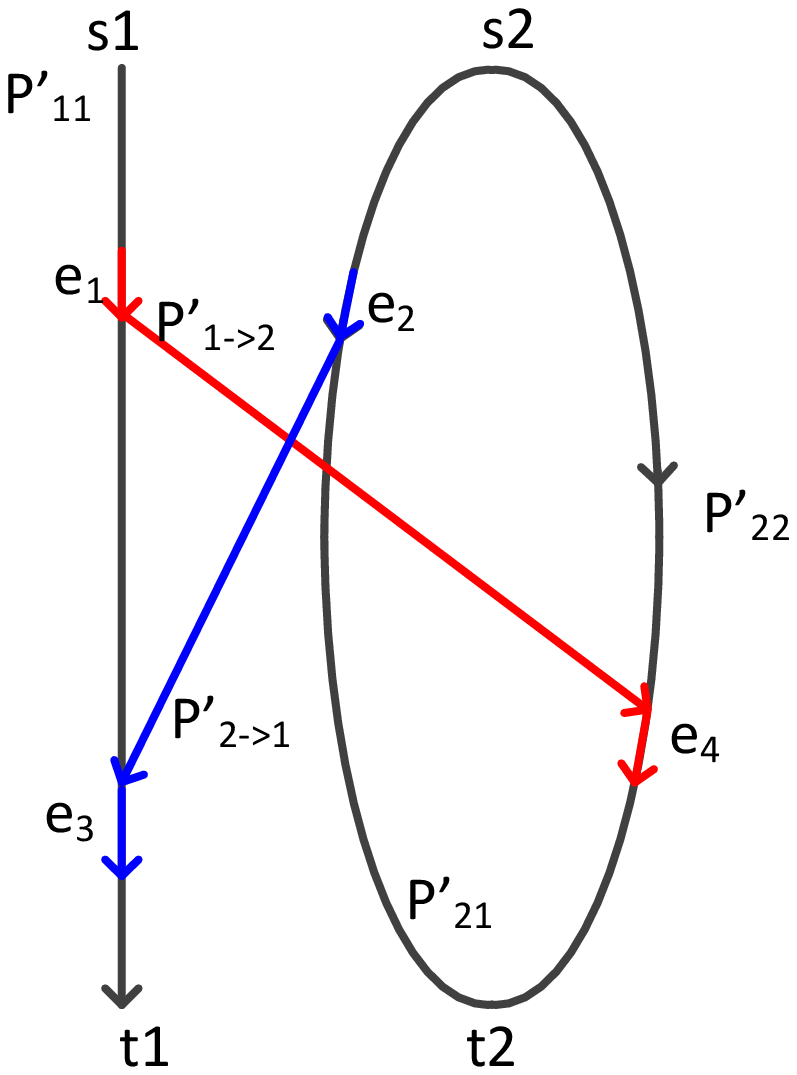}\vspace{-2mm}}
\subfigure[]{\label{fig:t2ex13}
\includegraphics[width=45mm,clip=false, viewport=0 10 270 300]{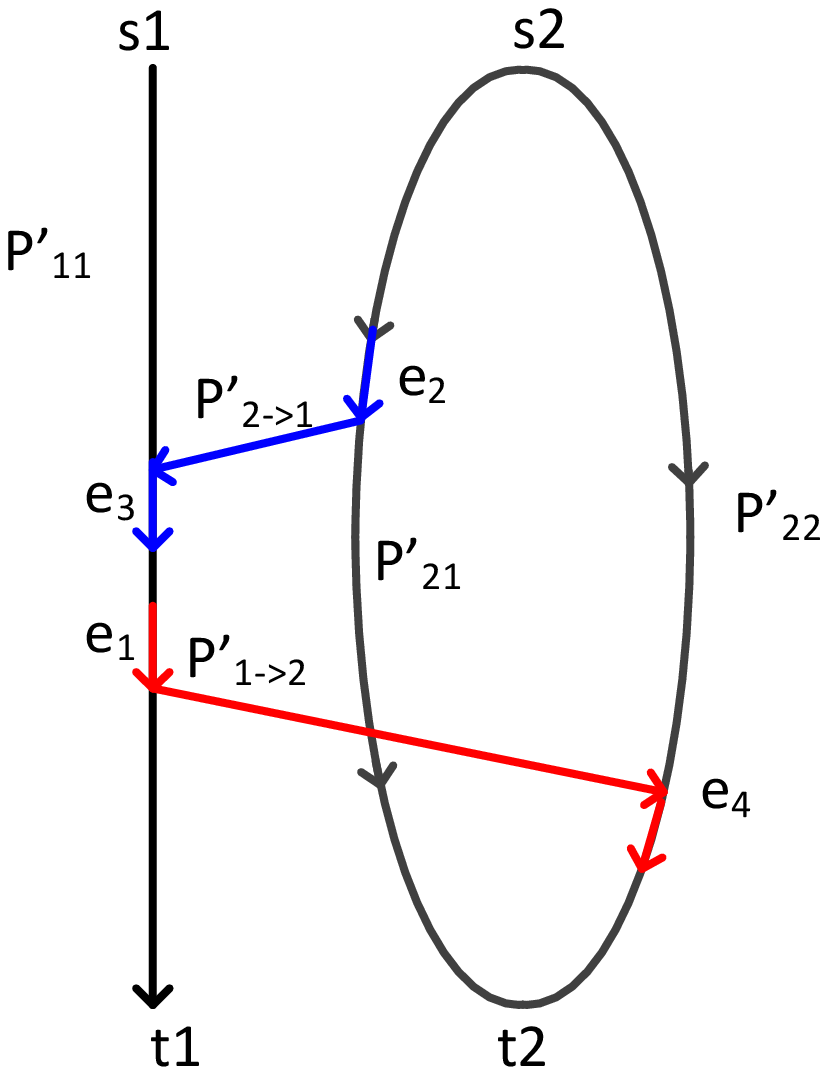}\vspace{-2mm}}
\caption{Possible subgraphs $G'$ when $P'_{11}$ does not overlap with either $P'_{21}$ or $P'_{22}$.}}
\vspace{-0.2in}
\end{figure*}
We now present our schemes for the different possibilities for $G'$.
For the class of $G'$ that fall in Fig. \ref{fig:t2ex11}, it suffices to use the approach in the proof of Theorem \ref{th:case1}. Namely, we use random linear network coding in the network and precoding at sources $s_2$ and $s_3$. As in this case $M_{21} \neq 0$, one needs to argue that $rank[M_{21}~M_{22}\xi] = 2$. Following the line of argument used previously, we can do this by demonstrating a choice of local coding coefficients such that $[\beta_1~\beta_2] = [1 ~0]$ and $[M_{21}~~M_{22}]=\left[{\begin{array}{ccc} 1&1&0\\ 0&0&1\\ \end{array}}\right]$.
%
However, such an approach does not work when the subgraph $G'$ belong to the class of graphs shown in Figs. \ref{fig:t2ex12} and \ref{fig:t2ex13}. For instance, it is easy to observe that if we use random coding on Fig. \ref{fig:t2ex12}, and precoding to cancel the $X_2$ component at $t_1$, then $t_2$ will receive a linear combination of $X_1$ and $X_2$ w.h.p., i.e., decoding $X_2$ at $t_2$ will fail.
Accordingly, when $G'$ looks like Fig. \ref{fig:t2ex12} or \ref{fig:t2ex13}, we require a different scheme that we now present.

\noindent \underline{{\it  Modified random coding for cases in Fig \ref{fig:t2ex12} and Fig \ref{fig:t2ex13}.}}\\
It is clear that the strategy of random linear network coding and precoding at the sources fails since the determinant of the matrix $[M_{21}~M_{22}\underline{\xi}]$ is identically zero for the cases in Fig. \ref{fig:t2ex12} and \ref{fig:t2ex13}. Thus, at the top level our approach is to modify the original graph $G$ by removing certain edges and identifying a special node in $G$ that is upstream of $t_2$. The transfer matrix on the two incoming edges of this special node can be expressed as $[\tilde{M}_{21}~\tilde{M}_{22}~\tilde{M}_{23}]$ such that the determinant of $[\tilde{M}_{21}~\tilde{M}_{22}\underline{\xi}]$ is not identically zero. Thus, at this node it becomes possible to remove the effect of $X_1$ via deterministic coding. Accordingly, our strategy is to first perform random linear coding at all nodes except the special node and those that are downstream of the special node. Following this, we perform deterministic coding at the special node to cancel the effect of $X_1$, and random linear coding downstream of it. Finally, we argue based on the precoding constraints that each terminal can decode its desired message. In the discussion below we outline each of the steps and the corresponding analysis in a systematic manner.

Recall that based on $G'$ (which is a subgraph of $G$) we have identified paths $P'_{11}$, $P'_{21}$, $P'_{22}$ that are all vertex disjoint, paths $\pm$ and $\pn$ and edges $e_1, \dots, e_4$. At the outset we demonstrate that certain structures in $G$, need not be considered. In particular, 
\begin{itemize}
\item if in $G$, there exists a path from $s_1$ to $t_1$ that has an overlap with $P'_{21} \cup P'_{22}$, it is clear that an alternate minimal subgraph $G{''}$ can be found that satisfies the conditions of Case 1.
\item In $G$, a path from $s_1$ cannot have an overlap with $path(e_2 - e_3)$. To see this note that $G'$ is a subgraph of $G$; therefore if path$(e_2 - e_3$) exists in it, then it necessarily has to belong to a path $P_{3i}$ from $s_3$ to $t_3$. We emphasize that the entire path including $e_2$ and $e_3$ have to belong to $P_{3i}$ because by assumption all nodes in the graph have in-degree + out-degree at most 3. In a similar manner, the path from $s_1$ that overlaps with path$(e_2 - e_3$) also needs to belong to path $P_{3j}$.If $i = j$, then it implies the existence of a path from $s_1$ to $t_1$ that has an overlap with $P_{21}' \cup P_{22}'$; however, this is explicitly ruled out by the discussion in the previous bullet. Thus, $i \neq j$; however, this is impossible since the paths $P_{3i}$ and $P_{3j}$ are edge disjoint.

\end{itemize}
Accordingly, in the discussion below, we will assume that the above scenarios do not occur.
\\
\noindent {\it Graph modification procedure for original graph $G$:}
\begin{itemize}
\item[(i)] Remove all edges downstream of $e_2$ on $P'_{21}$ that have no overlap with a path from $\cup_{i=1}^{5} P_{3i}$.
\item[(ii)] Identify an edge, denoted $e_{first}$ on $P'_{22}$, with the property that $e_{first}$ is the edge closest to $s_2$ such that there exists a $path(s_1 - e_{first})$.   Note that $e_{first}$ exists due to the existence of path $\pm$ in $G$.\
\item[(iii)] Remove edges downstream of $e_{first}$ while maintaining the following properties - (a) there exists a path from $e_{first} - t_2$, and (b) $max-flow(s_3 - t_3) = 5$. Rename $P'_{22}$ to be $path(s_2 - e_{first} - t_2)$. It is important to note that after this procedure, removal of any edge downstream of $e_{first}$ would cause either property (a) or (b) to fail.
\item[(iv)] Identify edge $e_{last} \in P'_{22}$ such that it is the edge closest to $t_2$ with the property that it has two incoming edges - $e'_1 \notin P'_{22}$ such that there exists $path(s_1 - e'_1)$ and $e'_2 \in P'_{22}$. Again $e'_1$ is guaranteed to exist as $\pm$ exists in $G$.
\end{itemize}
As a consequence of the modification procedure,  there is no overlap between $path(s_1 - e'_1)$ and $P'_{22}$. To see this, assume otherwise, i.e., an overlap segment, denoted $E_{os}$ exists between $path(s_1 - e'_1)$ and $P'_{22}$. As $e_{first}$ is the edge closest to $s_2$ such that there is a path between $s_1$ and $e_{first}$, it follows that $E_{os}$ is downstream of $e_{first}$ along $P'_{22}$. However, this contradicts the property of the modified graph after Step (iii) in the modification procedure above.

Next, note that $path(e_2 - e_3)$ has to overlap with a path from $\cup_{i=1}^5 P_{3i}$ (as $G$ is minimal) which means that the downstream neighboring edge of $e_2$ along $P'_{21}$ cannot belong to any path in $\cup_{i=1}^5 P_{3i}$ and will be removed in Step (i). Likewise 
the incoming edge of $t_2$ along $P'_{21}$ will also be removed.
At the end of the graph modification procedure, and using the observations made above, it is clear that we can identify a subgraph $\tilde{G}$ of $G$ that is topologically equivalent to either
Fig. \ref{fig:t2ex12_de} or \ref{fig:t2ex13_de}.

Next, we perform random linear coding over the modified graph except at edge $e_{last}$ and all the edges downstream of $e_{last}$, and impose the precoding constraints $[\beta_1 ~ \beta_2]\underline{\xi} = 0$ and $\underline{\gamma}^T \underline{\theta} = 0$. This ensures that $t_1$ is satisfied. Furthermore, note that there is no path from $e_{last}$ to $t_1$; therefore any code assignment on $e_{last}$ and its downstream edges will not affect decoding at $t_1$.

For $t_2$ to decode $X_2$, we first demonstrate that by using deterministic coding for edge $e_{last}$, the $X_1$ component can be canceled while the $X_2$ component can be maintained on $e_{last}$. Note that $e'_1$ and $e'_2$ denote the incoming edges of $e_{last}$; we denote the transfer matrix to these two edges by $[\tilde{M}_{21}~\tilde{M}_{22}~\tilde{M}_{23}]$.
\begin{claim}
\label{claim:tilde_G}
For the network structures in Fig. \ref{fig:t2ex12_de} and Fig. \ref{fig:t2ex13_de}, the determinant of $[\tilde{M}_{21}~\tilde{M}_{22}\underline{\xi}]$ is not identically zero where $\underline{\xi}$ satisfies $[\beta_1 ~ \beta_2]\underline{\xi} = 0$.
\end{claim}
\begin{proof}
Based on previous arguments, we have identified the subgraph $\tilde{G}$ of $G$ that is topologically equivalent to either Fig. \ref{fig:t2ex12_de} or \ref{fig:t2ex13_de}.
By Lemma \ref{lemma:common}, proving the claim is equivalent to showing that the determinant of eq. (\ref{eq:deterin}) is not identically zero. Based on $\tilde{G}$ it is evident that local coding vectors for the case of Fig. \ref{fig:t2ex12_de} can be chosen such that
\begin{figure*}[htbp]
\hspace{-0.1in}\center{
\subfigure[]{\label{fig:t2ex12_de}
\includegraphics[width=42mm,clip=false, viewport=0 5 270 300]{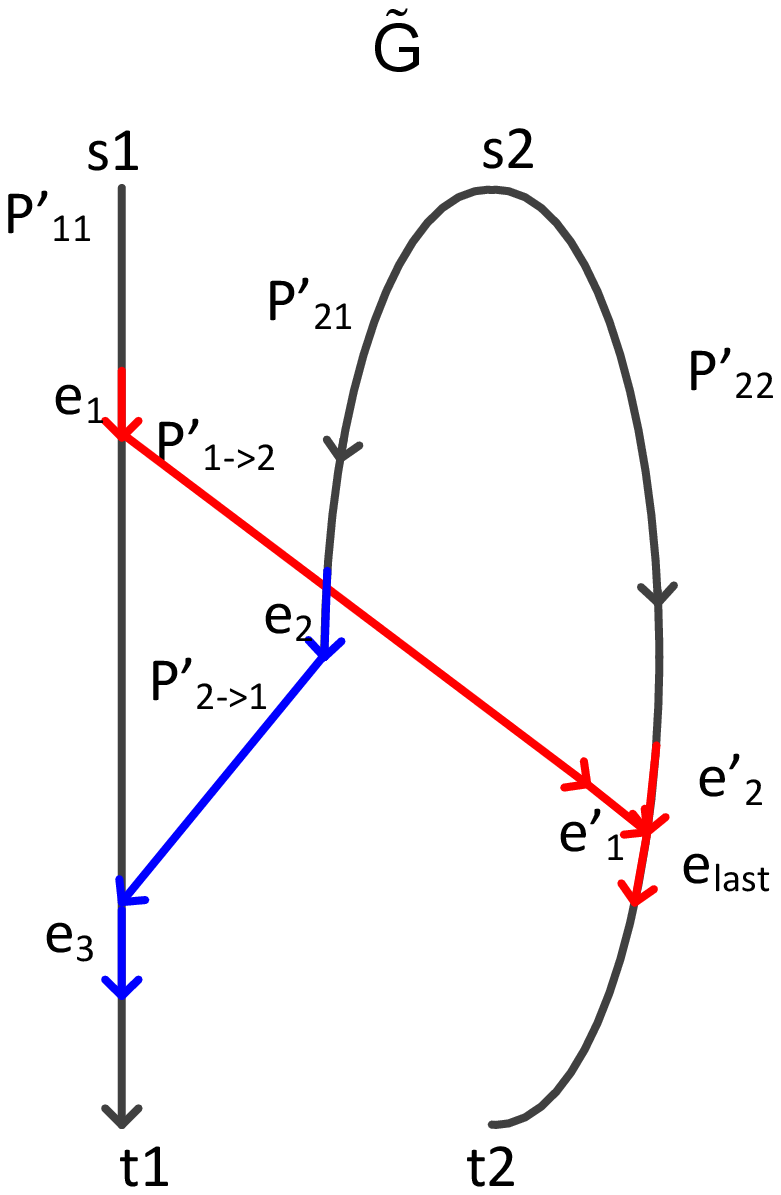}\vspace{-2mm}}
\subfigure[]{\label{fig:t2ex13_de}
\includegraphics[width=42mm,clip=false, viewport=0 5 270 300]{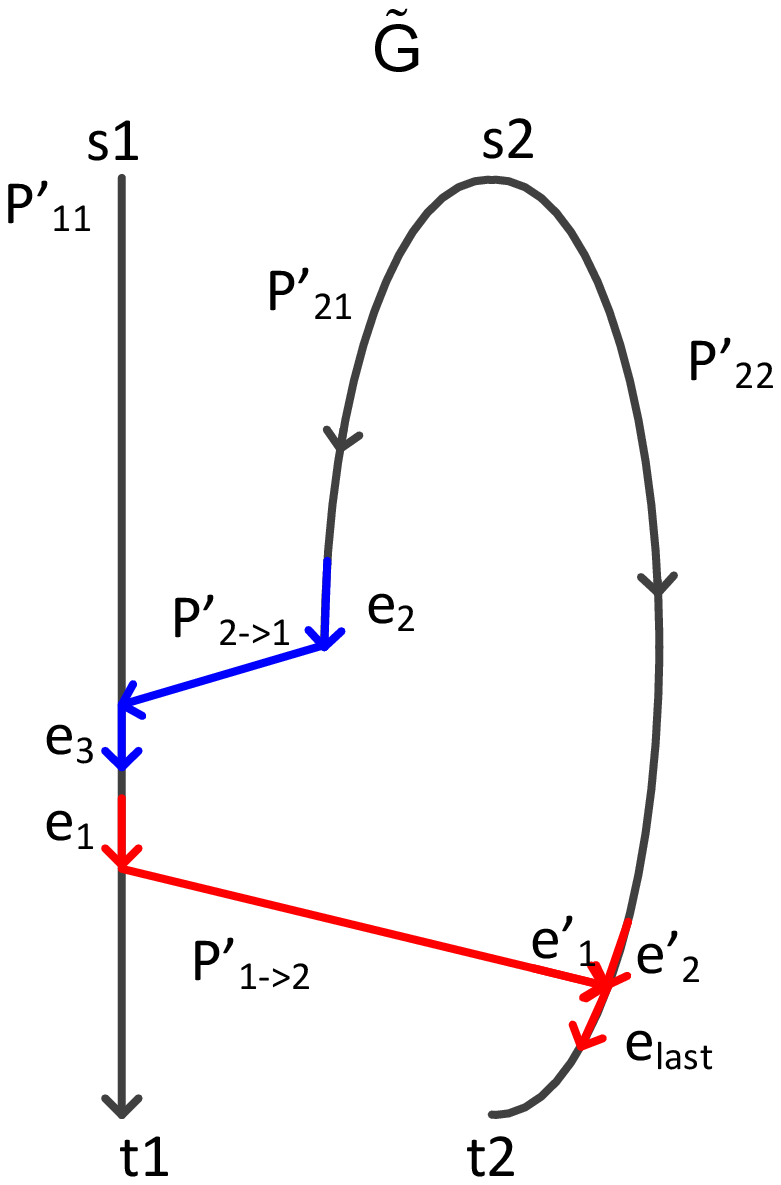}\vspace{-2mm}}
\subfigure[]{\label{fig:t2ex12_de_final}
\includegraphics[width=50mm,clip=false, viewport=0 5 310 340]{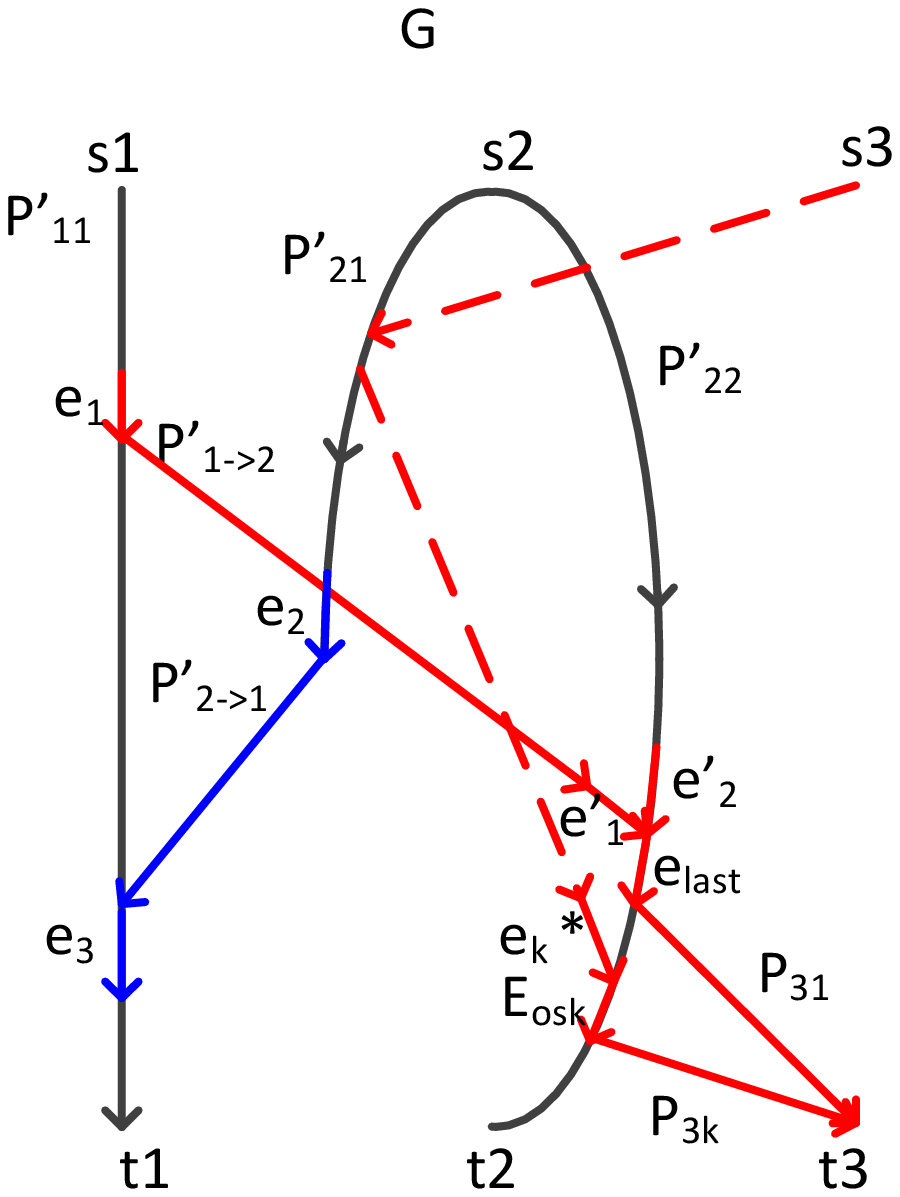}\vspace{-2mm}}
\caption{Figures (a) and (b) denote possible subgraphs $\tilde{G}$ obtained after the graph modification procedure for $G$. Figure (c) shows an example of the overlap between the red $s_3-t_3$ paths and $P'_{22}$.}}
\vspace{-0.2in}
\end{figure*}

$$[\beta_1~~\beta_2]=[1~0], \text{~and}$$
\begin{equation}
[\tilde{M}_{21}~~\tilde{M}_{22}]=\left[{\begin{array}{ccc}
1&0&0\\
0&0&1\\
\end{array}}\right].
\end{equation}
Similarly, for the case of Fig. \ref{fig:t2ex13_de} they can be chosen as
$$[\beta_1~~\beta_2]=[1~0], \text{~and}$$
\begin{equation}
[\tilde{M}_{21}~~\tilde{M}_{22}]=\left[{\begin{array}{ccc}
1&1&0\\
0&0&1\\
\end{array}}\right].
\end{equation}
Substituting the local coefficients into eq. (\ref{eq:deterin}) we have the required conclusion. 
\end{proof}
We now want to argue that $t_2$ can be satisfied.
Note that edge $e'_1$ must belong to a path from $\mathcal{P}_3$, as the graph is minimal. Assume that there are $k$ paths from $\mathcal{P}_3$ that overlap with $path(e_{last} - t_2)$; w.l.o.g. we assume that these are the paths $P_{31}, \dots, P_{3k}$.

Next, we note that there can be at most one overlap between a path $P_{3j}$ and $path(e_{last} - t_2)$. This is due to Step (iii) of the graph modification procedure, where we removed edges downstream of $e_{first}$, (and hence $e_{last}$) such that the $max-flow(s_3 - t_3) = 5$ and there is path between $e_{first} - t_2$. If there are multiple overlaps between $P_{3j}$ and $path(e_{last} - t_2)$, this would mean that there exists at least one edge that was not removed by Step (iii). As depicted in Fig. \ref{fig:t2ex12_de_final}, we denote the overlap segments as $E_{os1}, \dots, E_{osk}$, where $E_{osj}$ is upstream of $E_{os(j+1)}$ for $j = 1 ,..., k-1$ along $P'_{22}$. Also note that the first edge of $E_{os1}$ is $e_{last}$.

The next step in the code assignment is to use deterministic local coding coefficients so that the transmitted symbol on $e_{last}$ does not have an $X_1$ component. Note that it is guaranteed to have an $X_2$ component by the Claim \ref{claim:tilde_G} above. Following this, we again use random linear coding on edges downstream of $e_{last}$. By the definition of $e_{last}$ there is no edge $\in P'_{22}$ downstream of $e_{last}$ that is reachable from $s_1$. Thus all coding vectors along $P'_{22}$ downstream of $e_{last}$ do not have an $X_1$ component. Let the coding vector on the edge $\in E_{osk}$ closest to $t_2$ be denoted by $[0~|~\underline{\hat{\beta}}^T~|~\underline{\hat{\gamma}}^T]$, where it is evident that $\hat{\beta} \neq 0$ w.h.p. We enforce the precoding constraint $\underline{\hat{\gamma}}^T \underline{\theta} = 0$. This satisfies $t_2$.

Finally, we discuss the decoding at $t_3$. Consider the overlap segments  $E_{os1}, \dots, E_{osk}$ discussed above. Each of these overlap segments has an incoming edge that does not lie on $P'_{22}$ (the other has to be on $P'_{22}$). We denote these edges by $e^*_i, i = 1, \dots, k$, where we emphasize that $e^*_1 = e'_1$. Let the edges entering $t_3$ on paths $P_{3(k+1)}, \dots, P_{35}$ be denoted $e^*_{k+1}, \dots, e^*_5$. Denote the transfer matrix on the edges $e^*_1, \dots, e^*_5$ by $[\hat{M}_{31} ~|~ \hat{M}_{32} ~|~ \hat{M}_{33}]$. Note that with high probability it holds that $rank(\hat{M}_{33}) = 5$, since the max-flow from $s_3$ to these set of edges is $5$.

Next consider the rank of the coding vectors on edges $\{e_{last}, e^*_2, e^*_3, e^*_4, e^*_5\}$. For the sake of argument suppose that we remove the row of $\hat{M}_{33}$ corresponding to $e^*_1$ and replace it with the corresponding row of $e_{last}$. As we used a deterministic code assignment for edge $e_{last}$ the rank of the updated $\hat{M}_{33}$ may drop to four, however it will be no less than four since it has four linearly independent row vectors.

It can be seen that further random linear coding downstream of $e_{last}$ will therefore be such that $rank(M_{33})$ (recall that $[M_{31} | M_{32} |M_{33}]$ is the transfer matrix to $t_3$) is at least four w.h.p. Moreover, it can be seen that the information on $E_{osk}$ also reaches $t_3$, thus $t_3$ can decode $X_2$. Therefore at $t_3$ over the other four incoming edges we have a system of equations specified by the matrix $[\breve{M}_{31} | \breve{M}_{33}]$ (of dimension $4 \times 6$) with unknowns $X_1$ and $X_3$. Furthermore $rank(\breve{M}_{33}) \geq 3$. The constraints on $\underline{\theta}$ thus far dictate that there are $q^3 - 1$ non-zero choices for it. As shown in the appendix (cf. Lemma \ref{lemma:distinctValue}) this implies that there are at least $q^2 - 1$ distinct values for $\breve{M}_{33} \underline{\theta}$. For decoding $X_3$ at $t_3$, from Lemma \ref{lemma:partialDecode}, we need to have
\begin{equation}
\label{eq:cons32}
\breve{M}_{33}\underline{\theta}\notin span(\breve{M}_{31}).
\end{equation}
As there are at most $q$ vectors in the span of $M_{31}$, it follows that there are at least $q^2 - q - 1 > 0$ non-zero values of $\underline{\theta}$ such that $t_3$ can be satisfied.
\end{proof}

\begin{figure*}[t]
\hspace{-0.1in}\center{
\subfigure[]{\label{fig:col}
\includegraphics[width=35mm,clip=false, viewport=50 0 350 500]{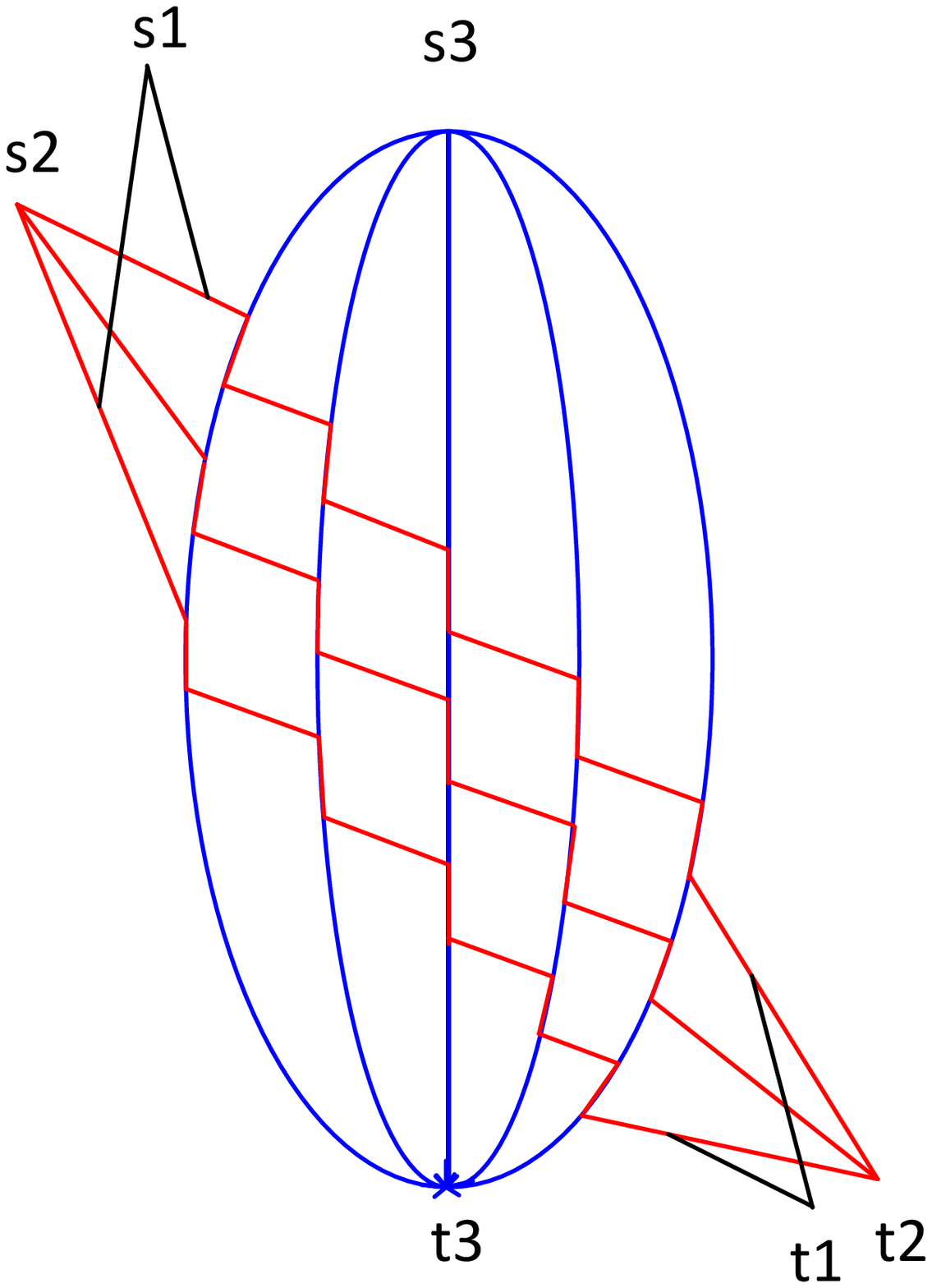}\vspace{-2mm}}
\subfigure[]{\label{fig:high_o}
\includegraphics[width=35mm,clip=false, viewport=0 0 440 720]{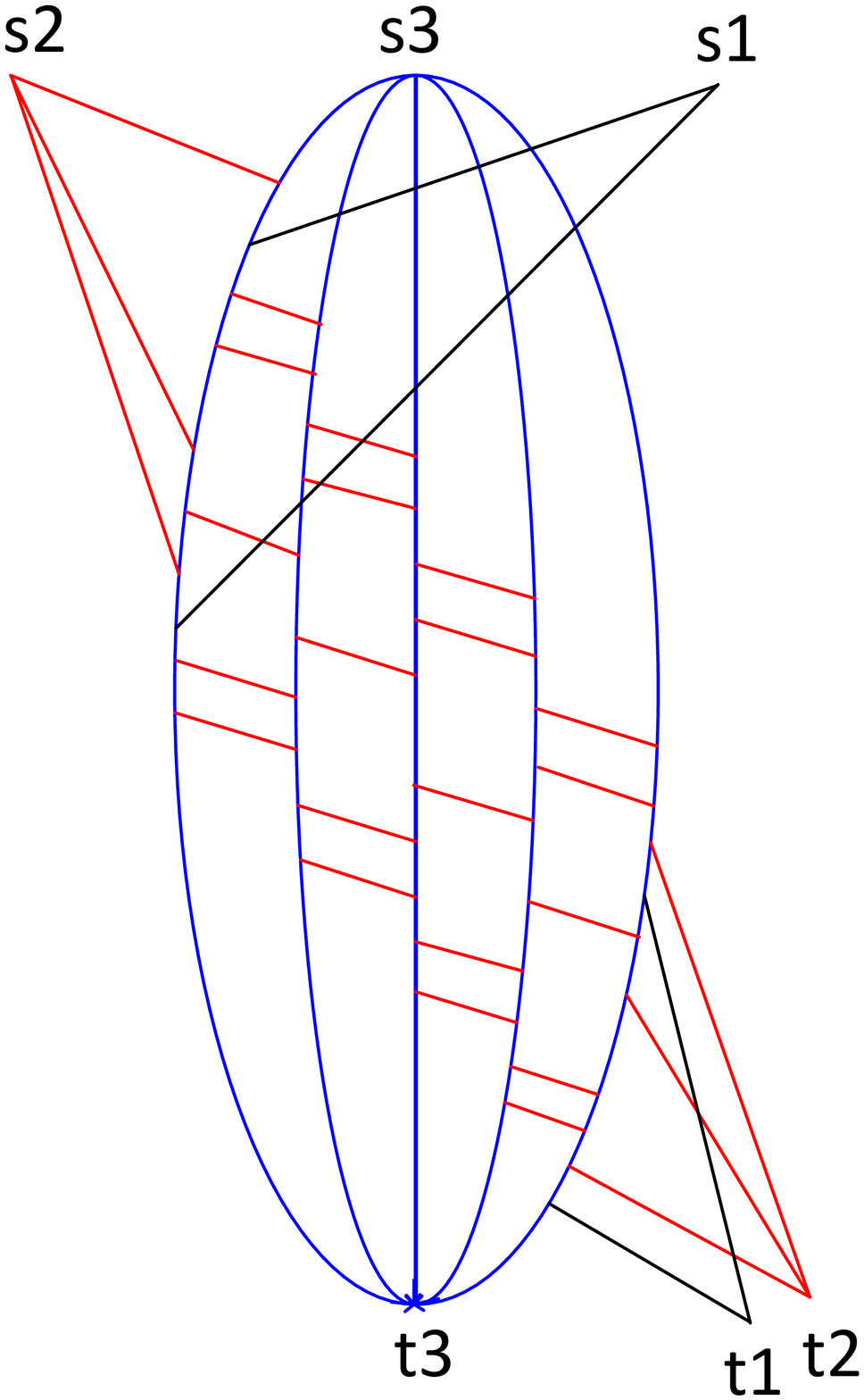}\vspace{-2mm}}
\subfigure[]{\label{fig:med_o}
\includegraphics[width=35mm,clip=false, viewport=0 0 440 720]{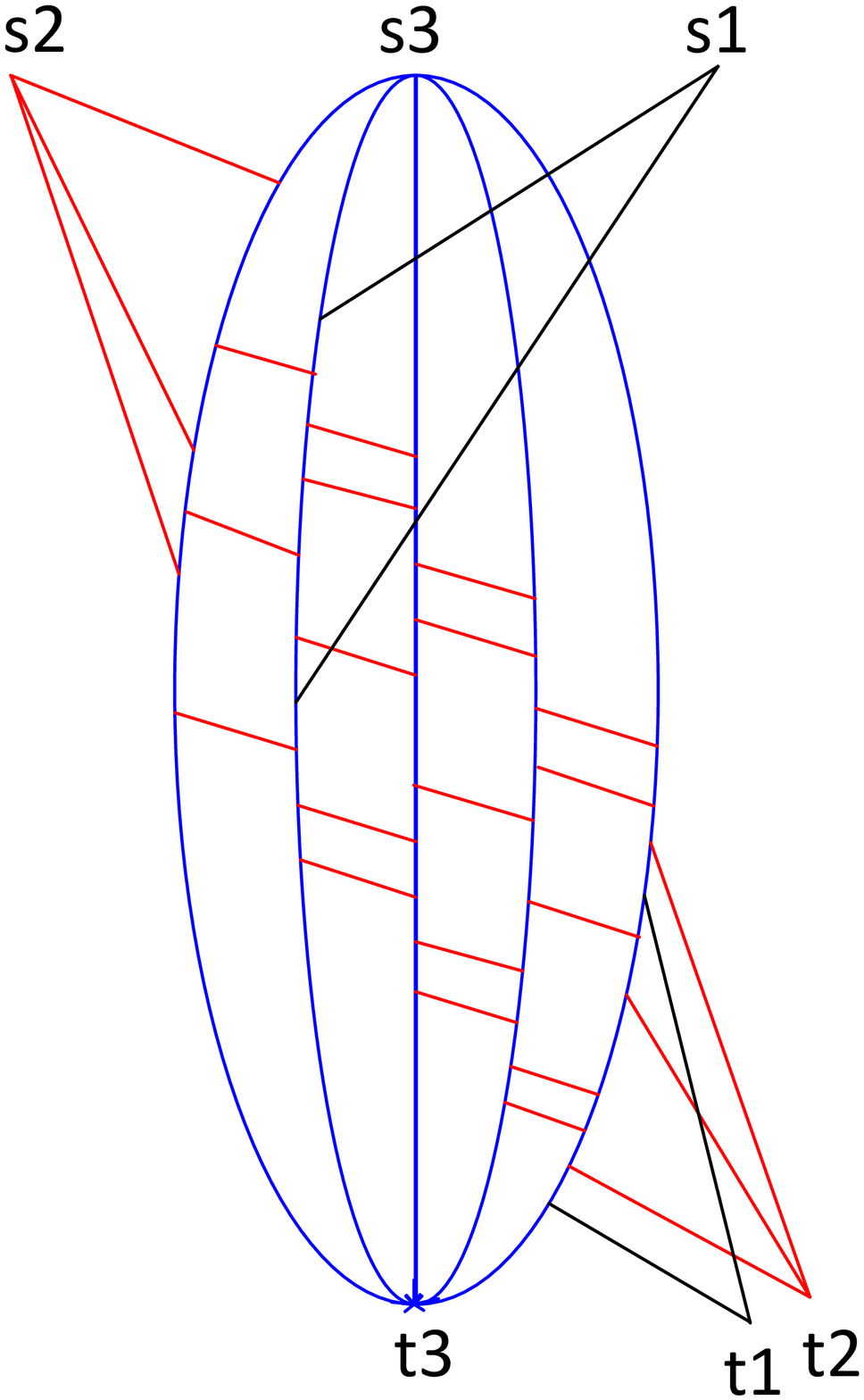}\vspace{-2mm}}
\subfigure[]{\label{fig:low_o}
\includegraphics[width=35mm,clip=false, viewport=0 0 440 720]{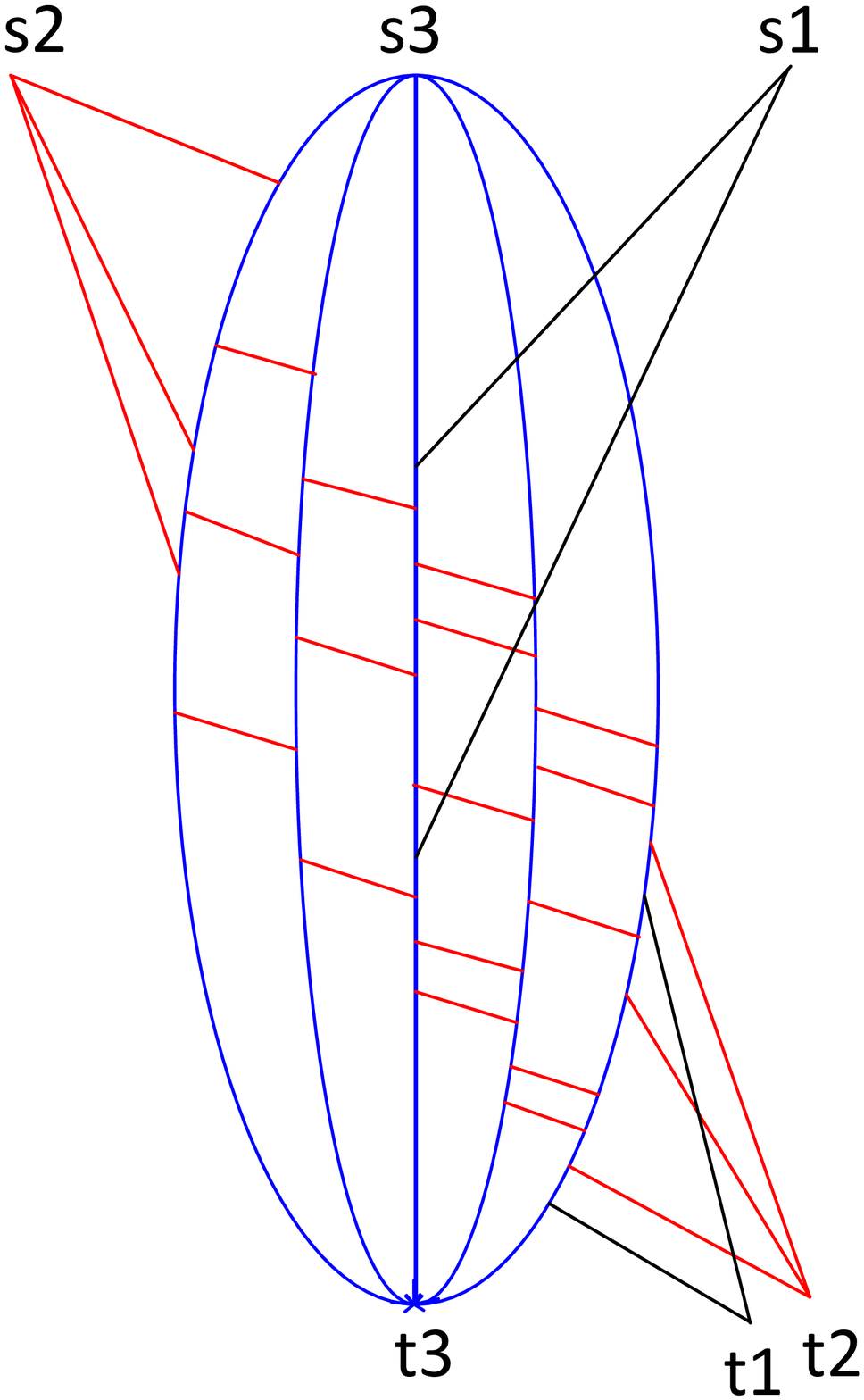}\vspace{-2mm}}
\caption{\label{fig:basis} a) Level-1 network. b) Level-2 network. c) Level-3 network. d) Level-4 network.}}
\vspace{-0.2in}
\end{figure*}
\section{Simulation Results}
\label{sec:simu}
Our feasibility results thus far have been for the case of unit-rate transmission over networks with unit-capacity edges. In this section, we present simulation results that demonstrate that these can also be used for networks with higher edge capacities, that can potentially support higher rates for the connections. The main idea is to pack multiple basic feasible solutions along with fractional routing solutions to achieve a higher throughput. The packing can be achieved by formulating appropriate integer linear programs. We compared these results to the case of solutions that can be achieved via pure fractional routing.

We applied our technique to several classes of networks. We did not see a benefit in the case of networks generated using random geometric graphs (this is consistent with previous results \cite{Traskov06}). We have found that our techniques are most powerful for networks where the paths between the various $s_i - t_i$ pairs have significant overlap. Accordingly, we experimented with four classes of networks (shown in Fig. \ref{fig:basis}) with varying levels of overlap between the different source-terminal pairs. 
The level-1 network (Fig. \ref{fig:basis}(a)) has the maximum overlap between the $s_1 - t_1$ paths and the other paths; the overlap decreases with an increase in the level number of the network.
The edge capacities in the networks were chosen randomly and independently with distributions as explained below. We conducted two sets of simulations. 
\begin{list}{}{\leftmargin=0.1in \labelwidth=0cm \labelsep = 0cm}
\item[$\bullet$~] {\it Simulation 1.} Let $C$ denote the edge capacity. For the level-1 network for the black edges we chose $P(C = 1) = 0.25, P(C=2) = 0.4, P(C=3) = 0.35$; for the other edges, $P(C = 1) = 0.15, P(C=2) = 0.6, P(C=3) = 0.25$. In the other networks we chose $P(C = 1) = 0.15, P(C=2) = 0.6, P(C=3) = 0.25$ for all the edges. Thus in this set of simulations, the maximum edge capacity is three. We generated 300 networks from these distributions and compared the performance of our schemes with pure fractional routing. The results shown in the first row of Table \ref{table:dis1} indicate that the level-1 network has the maximum number of instances where a difference in the throughput was observed; both $[1~2~5]$ and $[2~2~4]$ structures appear here. For the other networks, the $[2~2~4]$ structure appeared most often. The second row of Table \ref{table:dis1} records the average performance improvement when there was a difference between our scheme and routing; it varies between 4.9\% to 5.59\%.
\item[$\bullet$~] {\it Simulation 2.} In this set of simulations we increased the average edge capacity. For the level-1 network for the black edges we chose $P(C = 5) = 0.25, P(C=6) = 0.4, P(C=7) = 0.35$; for the other edges, $P(C = 5) = 0.15, P(C=6) = 0.6, P(C=7) = 0.25$. In the other networks we chose $P(C = 5) = 0.15, P(C=6) = 0.6, P(C=7) = 0.25$ for all the edges. Again, we generated 300 networks from these distributions and compared the performance of our schemes with pure fractional routing. The results shown in the third row of Table \ref{table:dis1} indicate that in this higher capacity simulation, the number of networks where our schemes outperform pure routing is significantly higher. For instance for the level-2 and level-3 networks more than 50\% of the networks showed an increase in the throughput using our methods. Another interesting point, is that one observes an increased gap for level-3 networks compared to the other cases. The fourth row of Table \ref{table:dis1} records the average performance improvement when there was a difference between our scheme and routing; it varies between 0.45\% to 1.16\%.
\end{list}

We found that though there were instances of all the structures being packed by the ILP, the majority were $[2~2~4]$ structures. For the level-4 network, since $[2~2~4]$ structure cannot be packed effectively, there is a significant drop in the proportions of networks that exhibit a difference with respect to routing as compared to the level-3 and level-4 networks. There were significant advantages in our approach for the case of networks with higher edge capacities as in these networks the chance of packing our basic feasible structures is higher. The average performance improvement obtained when there was a difference between our schemes and routing is not very high. We remark that the complexity of running the ILP increases with higher edge capacities and that was a limiting factor in our experiments; the performance improvement may be higher for large scale examples. Overall, our results indicate that there is a benefit to using our techniques even for networks with higher capacities, where the different source-terminal paths have a large overlap.
%

\begin{table}\scriptsize \caption{Proportions of networks with differences and performance improvement} \centering
{\begin{tabular}{c|c c c c}
\hline\\
{\bf Network} & Level-1  & Level-2 & Level-3 & Level-4 \\
\hline\\
Simulation 1 proportions& 5.33\% & 2.33\% & 1\% & 0\\
\hline\\
Performance improvement& 5.59\% & 5.06\% & 4.90\% & -\\
\hline\\
Simulation 2 proportions& 47\% & 53\% & 80.67\% & 2.33\%\\
\hline\\
Performance improvement& 1.16\% & 1.31\% & 1.36\% & 0.45\%\\
\hline
\end{tabular}}
\label{table:dis1} \vspace{-0.1in}
\end{table} \normalsize

%


\section{Conclusions and Future Work}
\label{sec:con}
In this work we considered the three-source, three-terminal multiple unicast problem for directed acyclic networks with unit capacity edges. Our focus was on characterizing the feasibility of achieving unit-rate transmission for each session based on the knowledge of the connectivity level vector. For the infeasible instances we have demonstrated specific network topologies where communicating at unit-rate is impossible, while for the feasible instances we have designed constructive linear network coding schemes that satisfy the demands of each terminal.
Our schemes are non-asymptotic and require vector network coding over at most two time units. Our work leaves out one specific connectivity level vector, namely $[1~2~4]$ for which we have been unable to provide either a feasible network code or a network topology where communicating at unit rate is impossible. Our experimental results indicate that there are benefits to using our techniques even for networks where the edges have higher and potentially different capacities. Specifically, our basic feasible solutions can be packed along with routing to obtain a higher throughput. Future work would include a study of real-world networks where these techniques are most useful.


\appendix

\begin{claim}
\label{claim:endecode}
For two independent random variables $X_1$ and $X_2$ with $H(X_1)=a$ and $H(X_2)=b$, if $H(X_1|Y)\leq \ep$ where $Y$ is another random variable with $H(Y)\leq a$, then $b-\ep\leq H(X_2|Y)\leq b$, $H(Y|X_2)\geq a-2\ep$.
\end{claim}
\begin{proof}
Since $H(X_1)=a$ and $H(X_1|Y)\leq \ep$, we have
$$
H(Y)=H(X_1,Y)-H(X_1|Y)\geq H(X_1)-H(X_1|Y)\geq a-\ep.
$$
Next $H(Y)\leq a$ implies that
$$
H(Y|X_1)=H(X_1| Y)+H(Y)-H(X_1)\leq \ep+a-a=\ep.
$$
As $X_1$ and $X_2$ are independent and $H(X_2)=b$, we have
\begin{equation*}
\begin{split}
b&=H(X_2)=H(X_2|X_1)\leq H(X_2|X_1, Y)+H(Y|X_1)\\
&\leq H(X_2|X_1, Y)+\ep\leq H(X_2|Y)+\ep\leq b+\ep.
\end{split}
\end{equation*}
Thus,
\vspace{-0.1in}
\begin{equation*}
b-\ep\leq H(X_2|Y)\leq b.
\end{equation*}
Finally, we obtain
\begin{equation*}
\begin{split}
H(Y|X_2)&=H(Y)-I(Y;X_2)=H(Y)+H(X_2|Y)-H(X_2)\\
&\geq a-\ep+b-\ep-b=a-2\ep
\end{split}
\end{equation*}
\end{proof}

\begin{lemma}
\label{lemma:common}
If $\beta_1\neq 0$, $det([M_{21}~~M_{22}\underline{\xi}])$ can be represented by
\vspace{-0.13in}
\begin{equation}
\label{eq:deterin}
\frac{\xi_2}{\beta_1}\text{det}\left[{\begin{array}{cc}
\alpha'_1&-\beta_2\beta'_{11}+\beta_1\beta'_{12}\\
\alpha'_2&-\beta_2\beta'_{21}+\beta_1\beta'_{22}\\
\end{array}}\right].
\end{equation}
where $\underline{\xi}$ satisfies $[\beta_1 ~\beta_2] \underline{\xi} = 0$.
\end{lemma}

\begin{proof}
Because $\underline{\xi}$ satisfies $[\beta_1 ~\beta_2] \underline{\xi} = 0$, we can have $\xi_1=-\beta_2\xi_2/\beta_1$. Note $\xi_2$ can be selected to be nonzero, regardless of the value of $\beta_2$. 
By substituting $\xi_1$ into $[M_{21}~~M_{22}\underline{\xi}]$, the determinant of $[M_{21}~~M_{22}\underline{\xi}]$ becomes
\begin{equation}
\label{eq:12random}
\begin{split}
&\text{det}\left[M_{21}~M_{22}\left[{\begin{array}{c}
-\frac{\beta_2\xi_2}{\beta_1}\\
\xi_2\\
\end{array}}\right]\right] =\frac{\xi_2}{\beta_1}\text{det}\left[{\begin{array}{cc}
\alpha'_1&-\beta_2\beta'_{11}+\beta_1\beta'_{12}\\
\alpha'_2&-\beta_2\beta'_{21}+\beta_1\beta'_{22}\\
\end{array}}\right],
\end{split}
\end{equation}
%
where $\xi_2/\beta_1$ is nonzero.
\end{proof}
\begin{lemma}
\label{lemma:partialDecode} Consider a system of equations
$Z=H_1X_1+H_2X_2$, where $X_1$ is a vector of length $l_1$ and $X_2$
is a vector of length $l_2$ and $Z\in span([H_1~~H_2])$\footnote{$span(A)$ refers to
the column span of $A$.}. The matrix
$H_1$ has dimension $z_t\times l_1$, and rank $l_1-\sigma$, where
$0\leq\sigma\leq l_1$. The matrix $H_2$ is full rank and has
dimension $z_t\times l_2$ where $z_t\geq (l_1+l_2-\sigma)$.
Furthermore, the column spans of $H_1$ and $H_2$ intersect only in
the all-zeros vectors, i.e. $span(H_1)\cap
span(H_2)=\{0\}$. Then there exists a unique solution for
$X_2$.
\end{lemma}

\begin{proof} Since $Z\in span([H_1~~H_2])$, there
exists $X_1$ and $X_2$ such that $Z=H_1X_1+H_2X_2$. Now assume there
is another set of $X'_1$ and $X'_2$ such that $Z=H_1X'_1+H_2X'_2$.
Then we will have
\begin{equation}
\label{eq:pf} H_1(X_1-X'_1)=H_2(X_2-X'_2).
\end{equation}
Because $span(H_1)\cap span(H_2)=\{0\}$, both sides of eq.
\ref{eq:pf} are zero. Furthermore, since $H_2$ is a full rank
matrix, $X_2=X'_2$. The solution of $X_2$ is unique.
\end{proof}

\begin{lemma}
\label{lemma:distinctValue} There are at least $q^2 - 1$ distinct values for $\breve{M}_{33} \underline{\theta}$ when there are $q^3 - 1$ distinct values for $\underline{\theta}$.
\end{lemma}

\begin{proof} Since $\breve{M}_{33}$ is a $4\times 5$ matrix with rank at least 3, we can find two vectors $\breve{\underline{\gamma}}_1$ and $\breve{\underline{\gamma}}_2$ such that the matrix
$\breve{M}'_{33}=[
\breve{M}_{33}^T~|~\breve{\underline{\gamma}}_1~|~\breve{\underline{\gamma}}_2]^T$
and $rank(\breve{M}'_{33})=5$. This implies that there are $q^3-1$ distinct values for $\breve{M}'_{33} \underline{\theta}$. Next note that since $rank(M_{33})\geq 4$,  $\breve{\underline{\gamma}}_1$ can be selected as the coding vector for $\underline{\theta}X_3$ on $E_{osk}$ so that
$rank[
\breve{M}_{33}^T~|~
\breve{\underline{\gamma}}_1]^T\geq 4$.
The precoding constraint implies that $\breve{\underline{\gamma}}_1^T \underline{\theta}=0$. Hence, by removing $\breve{\underline{\gamma}}_1\underline{\theta}$ from $\breve{M}'_{33} \underline{\theta}$, there continue to be $q^3-1$ distinct vectors. If we further remove $\breve{\underline{\gamma}}_2\underline{\theta}$ from $\breve{M}'_{33} \underline{\theta}$, there will be at least $q^2-1$ distinct values, i.e., there are $q^2-1$ distinct values for $\breve{M}_{33} \underline{\theta}$.
\end{proof}

\bibliographystyle{IEEEtran}
\bibliography{unicast}

\begin{thebibliography}{10}
\providecommand{\url}[1]{#1}
\csname url@samestyle\endcsname
\providecommand{\newblock}{\relax}
\providecommand{\bibinfo}[2]{#2}
\providecommand{\BIBentrySTDinterwordspacing}{\spaceskip=0pt\relax}
\providecommand{\BIBentryALTinterwordstretchfactor}{4}
\providecommand{\BIBentryALTinterwordspacing}{\spaceskip=\fontdimen2\font plus
\BIBentryALTinterwordstretchfactor\fontdimen3\font minus
  \fontdimen4\font\relax}
\providecommand{\BIBforeignlanguage}[2]{{%
\expandafter\ifx\csname l@#1\endcsname\relax
\typeout{** WARNING: IEEEtran.bst: No hyphenation pattern has been}%
\typeout{** loaded for the language `#1'. Using the pattern for}%
\typeout{** the default language instead.}%
\else
\language=\csname l@#1\endcsname
\fi
#2}}
\providecommand{\BIBdecl}{\relax}
\BIBdecl

\bibitem{rm}
R.~Koetter and M.~M\'{e}dard, ``{An Algebraic approach to network coding},''
  \emph{IEEE/ACM Trans. on Networking}, vol. 11, no. 5, pp. 782--795, 2003.

\bibitem{hoMKKESL06}
T.~Ho, M.~M\'{e}dard, R.~Koetter, D.~R. Karger, M.~Effros, J.~Shi, and
  B.~Leong, ``{A Random Linear Network Coding Approach to Multicast},''
  \emph{IEEE Trans. on Info. Th.}, vol. 52(10), pp. 4413--4430, 2006.

\bibitem{ebrahimiF11}
J.~B. Ebrahimi and C.~Fragouli, ``Algebraic algorithms for vector network
  coding,'' \emph{IEEE Trans. on Info. Th.}, vol.~57, no.~2, pp. 996 --1007,
  2011.

\bibitem{DoughertyFZ05}
R.~Dougherty, C.~Freiling, and K.~Zeger, ``Insufficiency of linear coding in
  network information flow,'' \emph{IEEE Trans. on Info. Th.}, vol.~51, no.~8,
  pp. 2745 -- 2759, 2005.

\bibitem{medardEHK03}
M.~M\'{e}dard, M.~Effros, T.~Ho, and D.~Karger, ``{On coding for non-multicast
  networks},'' \emph{$41^{st}$ Allerton Conference on Communication, Control
  and Computing}, 2003.

\bibitem{yan06isit}
X.~Yan, R.~W. Yeung, and Z.~Zhang, ``{The Capacity Region for Multi-source
  Multi-sink Network Coding},'' in \emph{IEEE Intl. Symp. on Info. Th.}, 2007,
  pp. 116--120.

\bibitem{HarveyIT}
N.~Harvey, R.~Kleinberg, and A.~Lehman, ``{On the capacity of information
  networks},'' \emph{IEEE Trans. on Info. Th.}, vol. 52, no. 6, pp. 2345--2364,
  2006.

\bibitem{Traskov06}
D.~Traskov, N.~Ratnakar, D.~Lun, R.~Koetter, and M.~Medard, ``{Network Coding
  for Multiple Unicasts: An Approach based on Linear Optimization},'' in
  \emph{IEEE Intl. Symp. on Info. Th.}, 2006, pp. 1758--1762.

\bibitem{tracy06allerton}
T.~Ho, Y.~Chang, and K.~J. Han, ``{On constructive network coding for multiple
  unicasts},'' in \emph{44th Allerton Conf. on Comm., Contr. and Comp.}, 2006.

\bibitem{wangIT10}
C.-C. Wang and N.~B. Shroff, ``{Pairwise Intersession Network Coding on
  Directed Networks},'' \emph{IEEE Trans. on Info. Th.}, vol. 56, no. 8, pp.
  3879--3900, 2010.

\bibitem{feder09}
E.~Erez and M.~Feder, ``{Improving the Multicommodity Flow Rate with Network
  Codes for Two Sources},'' \emph{IEEE J. Select. Areas Comm.}, vol. 27, no. 5,
  pp. 814--824, 2009.

\bibitem{Kamath11}
S.~U. Kamath, D.~N.~C. Tse, and V.~Anantharam, ``{Generalized Network Sharing
  Outer Bound and the Two-Unicast Problem},'' in \emph{Proceedings of Netcod},
  2011, pp. 1--6.

\bibitem{javidi08}
J.~Price and T.~Javidi, ``{Network Coding Games with Unicast Flows},''
  \emph{IEEE J. Select. Areas Comm.}, vol. 26, no. 7, pp. 1302--1316, 2008.

\bibitem{JafarISIT}
A.~Das, S.~Vishwanath, S.~A. Jafar, and A.~Markopoulou, ``{Network Coding for
  Multiple Unicasts: An Interference Alignment Approach},'' in \emph{IEEE Intl.
  Symp. on Info. Th.}, 2010, pp. 1878 -- 1882.

\bibitem{interfernce10}
A.~Ramakrishnan, A.~Das, H.~Maleki, A.~Markopoulou, S.~Jafar, and
  S.~Vishwanath, ``{Network coding for three unicast sessions: Interference
  alignment approaches},'' in \emph{48th Allerton Conf. on Comm., Contr. and
  Comp.}, 2010, pp. 1054 -- 1061.

\bibitem{LiLiunicast}
Z.~Li and B.~Li, ``{Network coding: the Case of Multiple Unicast Sessions},''
  in \emph{42st Allerton Conf. on Comm., Contr. and Comp.}, 2004.

\bibitem{wang11}
J.~Han, C.-C. Wang, and N.~Shroff, ``{Analysis of Precoding-based Intersession
  Network Coding and The Corresponding 3-Unicast Interference Alignment
  Scheme},'' in \emph{49th Allerton Conf. on Comm., Contr. and Comp.}, 2011,
  pp. 1033--1040.

\bibitem{kamalRLL11}
A.~E. Kamal, A.~Ramamoorthy, L.~Long, and S.~Li, ``Overlay protection against
  link failures using network coding,'' \emph{IEEE/ACM Trans. on Networking},
  vol.~19, no.~4, pp. 1071 --1084, 2011.

\bibitem{shizhengR11}
S.~Li and A.~Ramamoorthy, ``{Protection against link errors and failures using
  network coding in overlay networks},'' \emph{IEEE Trans. on Comm.}, vol.~59,
  no.~2, pp. 518--528, 2011.

\bibitem{shenvi}
S.~Shenvi and B.~K. Dey, ``{A Simple Necessary and Sufficient Condition for the
  Double Unicast Problem},'' in \emph{IEEE Intl. Conf. Comm.}, 2010, pp. 1--5.

\bibitem{huangR_ITW11}
S.~Huang and A.~Ramamoorthy, ``{An achievable region for the double unicast
  problem based on a minimum cut analysis},'' in \emph{IEEE Information Theory
  Workshop (ITW)}, 2011, pp. 120 -- 124.

\bibitem{caja08}
V.~Cadambe and S.~Jafar, ``Interference {alignment} and {degrees} of {freedom}
  of the {K} user {interference} {channel},'' \emph{IEEE Trans. on Info. Th.},
  vol.~54, no.~8, pp. 3425--3441, 2008.

\bibitem{ramamoorthyL12_jnl}
A.~Ramamoorthy and M.~Langberg, ``{Communicating the sum of sources over a
  network},'' 2012 [Online] available at http://arxiv.org/abs/1001.5319.

\bibitem{LSB06}
M.~Langberg, A.~Sprintson, and J.~Bruck, ``{The encoding complexity of network
  coding},'' \emph{IEEE Trans. on Info. Th.}, vol. 52, no. 6, pp. 2368--2397,
  2006.

\bibitem{motwaniR}
R.~Motwani and P.~Raghavan, \emph{Randomized Algorithms}.\hskip 1em plus 0.5em
  minus 0.4em\relax Cambridge University Press, 1995.

\end{thebibliography}
\end{document}